\documentclass[twocolumn,letter]{IEEEtran}
\usepackage{graphicx,fancybox,color}
\usepackage{latexsym,extarrows}
\usepackage{amsmath,amstext,amsthm,amssymb}

\DeclareMathOperator\erfc{erfc}

\def\Tr{\mathop{\rm Tr}\nolimits}
\def\d{\mathop{\rm d}  \hspace{-1.5pt}}

\def\argmin{\mathop{\rm argmin}}
\newcommand{\map}[1]{\mathcal{#1}}

\newcommand{\set}[1]{\mathsf{#1}}

\newcommand{\spc}[1]{\mathcal{#1}}

\def\>{\rangle}
\def\<{\langle}

{\catcode`\|=\active
  \gdef\Braket#1{\begingroup
\mathcode`\|32768\let|\BraVert\left<{#1}\right>\endgroup}}
\def\BraVert{\egroup\,\mid\,\bgroup}

\newtheorem{mythm}{Theorem}
\newtheorem{Lemma}{Lemma}

\def\R{\mathbb R}
\def\N{\mathbb N}
\def\C{\mathbb C}
\def\Z{\mathbb Z}

\definecolor{kmblue}{rgb}{0.19, 0.25, 0.91}
\definecolor{kmred}{rgb}{0.79, 0.29, 0.0}
\definecolor{kmgreen}{rgb}{0, 0.42, 0.24}

\begin{document}

\title{Compression for quantum population coding}
\author{Yuxiang Yang, Ge Bai, Giulio Chiribella, and Masahito Hayashi~\IEEEmembership{Fellow,~IEEE}
\thanks{Y.  Yang  (e-mail: yangyx09@connect.hku.hk) and G. Bai  (e-mail: baige@connect.hku.hk)  are with the Department of Computer Science, The University of Hong Kong, Pokfulam Road, Hong Kong, and with the HKU Shenzhen Institute of Research and Innovation, Kejizhong 2$^{\rm nd}$ Road,  Shenzhen, China.}
\thanks{G. Chiribella (e-mail: giulio.chiribella@cs.ox.ac.uk)   is with the Department of Computer Science, The University of Oxford, Parks Road, Oxford, UK,  with the Canadian Institute for Advanced Research, CIFAR Program in Quantum Information Science,  with the Department of Computer Science, The University of Hong Kong, Pokfulam Road, Hong Kong, and with the HKU Shenzhen Institute of Research and Innovation, Kejizhong 2$^{\rm nd}$ Road,  Shenzhen, China. }
\thanks{M. Hayashi   is with the Graduate School of Mathematics, Nagoya University,
Furocho, Chikusaku, Nagoya, 464-860, Japan,
and
Centre for Quantum Technologies, National University of Singapore, 3 Science Drive 2, Singapore 117542.  }  
}

\markboth{Y. Yang, G. Bai, G. Chiribella, and M. Hayashi: Compression for quantum population coding}{}

\maketitle

\begin{abstract}
We study the compression of  $n$ quantum systems, each prepared in the same state belonging to a given parametric family of quantum states.
For a family  of states with $f$ independent parameters, we devise an asymptotically faithful protocol that requires a hybrid memory of  size  $(f/2)\log n$, including both quantum and classical bits.  Our construction uses a quantum version of local asymptotic normality and, as an intermediate step, solves the problem of compressing  displaced thermal states of $n$ identically prepared modes.  In both cases, we show that  $(f/2)\log n$ is the minimum amount of  memory   needed  to achieve asymptotic faithfulness. In addition, we  analyze how much of the memory needs to be quantum.   We find that the ratio between quantum and classical bits can be made arbitrarily small, but cannot reach zero: unless all the quantum states in the family commute, no  protocol using only classical bits can be faithful, even if it uses an arbitrarily large number of classical  bits.

\end{abstract}


\section{Introduction}
Many problems in quantum information theory involve a source that prepares multiple  copies of the same quantum state.  This is the case, for example, of quantum tomography \cite{banaszek2013focus}, quantum cloning \cite{review-cloning,cloning-continuous}, and  quantum state discrimination \cite{barnett2009quantum}. The state prepared by the source  is generally unknown to the agent who has to carry out the task.  Instead, the agent knows that the state belongs to some parametric family of density matrices  $\{\rho_\theta\}_{\theta\in \Theta}$, with the parameter $\theta$ varying in the set $\Theta$.  It is generally assumed  that the source prepares each particle identically and independently:  when the source is used $n$ times, it generates $n$ quantum particles in the tensor product state $\rho_\theta^{\otimes n}$.

A fundamental question is how much information is contained in the $n$-particle state $\rho_\theta^{\otimes n}$.      One way to address the question is to quantify the minimum amount of  memory  needed to  store the state, or equivalently, the minimum amount of communication needed to transfer the state from a sender to a receiver.  Solving this problem requires an optimization over all possible compression protocols.

%When the number of copies is large, it is tempting to use a classical protocol, where  the parameter $\theta$ is estimated and the estimate is stored in a classical memory.   However, this type of storage is generally not faithful, as shown in \cite{yang2014certifying}  for several examples of pure state families. In order to achieve a faithful storage, a non-zero amount of quantum memory is generally required.

It  is important to stress that the problem of storing the $n$-copy states $\{\rho_\theta^{\otimes n}\}_{\theta \in  \Theta}$ in a quantum memory is different from the standard problem of quantum data compression
   \cite{schumacher1995quantum,jozsa1994anew,lo1995quantum}.
   %In standard quantum data compression, each use of the source generates a \emph{different} quantum state, drawn at random from a given probability distribution,  and the goal is to   reproduce with high probability  the sequence generated by the source.
    In our scenario, the mixed state $\rho_\theta$ is not regarded as the average state of an information source, but, instead, as a physical encoding of the parameter $\theta$.  The goal of  compression is to preserve the encoding of the parameter $\theta$, by storing the state $\rho_\theta^{\otimes n}$ into a memory and retrieving it with high fidelity for all possible values of $\theta$.  To stress the difference with standard quantum compression, we refer to our scenario as  \emph{compression for quantum population coding}. The expression ``quantum population coding" refers to the encoding of the parameter $\theta$ into the many-particle state $\rho_\theta^{\otimes n}$, viewed as the state of a ``population" of quantum systems.  We choose this expression in analogy with   the classical notion of population coding, where  a parameter  $\theta$ is encoded  into the   population of $n$ individuals  \cite{WAN}.  The typical example of population coding arises in computational neuroscience, where the population consists of neurons and the parameter  $\theta$ represents an external stimulus.

 The compression for quantum population coding   has been studied by Plesch and Bu\v zek \cite{buzek} in the case where $\rho_\theta$ is a pure qubit state and no error is tolerated (see also \cite{rozema} for a prototype experimental implementation).
%A prototype experiment based on this proposal was implemented by Rozema \emph{et al} \cite{rozema} on a photonic platform.
A first extension to mixed states, higher dimensions, and non-zero error was proposed by some of us in    \cite{yang-chiribella-2016-prl}.  The protocol therein was proven to be optimal under the assumption that the decoding operation must satisfy  a suitable conservation law.  Later, it was shown that, when the conservation law is lifted, a new protocol can achieve a better compression, reaching the ultimate information-theoretic bound set by Holevo's bound \cite{universal}. This result applies to two-dimensional quantum systems with completely unknown Bloch vector and/or completely unknown purity.  The classical version of  the compression for population coding was addressed  in \cite{classical}.  However,    finding the optimal protocol for  arbitrary parametric families and for quantum systems of arbitrary dimension  has remained as an open problem so far.
%Solving this problem  is  the goal of the present paper.

 In this paper, we provide  a general theory of compression for   quantum states   of the form $\rho_\theta^{\otimes n}$.
We consider two categories of states: {\em (i)} generic quantum states in finite dimensions, and  {\em (ii)}   displaced thermal states in infinite dimension.  These two categories of states are  connected by the quantum version of local asymptotic normality (Q-LAN) \cite{H-LAN,LAN,LAN2,kahn-thesis,LAN3}, which locally reduces the tensor product state $\rho_\theta^{\otimes n}$ to a displaced thermal state,  regarded as the quantum version of the normal distribution.

We  will discuss first the compression of  displaced thermal states.
Then, we  will employ Q-LAN to reduce the problem of compressing generic finite-dimensional states to the problem  of compressing displaced thermal states.
In both cases,  our compression protocol uses a hybrid memory, consisting both of classical and quantum bits. For a family of quantum states described by $f$ independent parameters, the total size of the memory is $f/2 \log n$ at the leading order, matching the ultimate limit set by  Holevo's bound \cite{holevo-1973}.
%The compression is thus independent of assumptions on the prior information and is more readily to be put into practice.
%The protocol presented in this work uses less memory compared to the one in \cite{} comparing their performances on the full model case.

An intriguing feature of our compression protocol is that the ratio between the number of quantum bits and the number of classical bits  can be made arbitrarily close to zero, but not exactly equal to zero.  Such a feature is not an accident: we show that, unless the states commute, every asymptotically faithful protocol must use a non-zero amount of quantum memory.  This result extends  an observation made  in \cite{yang2014certifying} from certain families of pure states to generic families of states.

The  paper is structured as follows. In section \ref{sec-result} we state the main results of the paper.  In Section \ref{sec-thermal} we study the compression of displaced thermal states.     In Section \ref{sec-main} we provide the  protocol for the compression of identically prepared finite-dimensional states.
In Section \ref{sec:classical} we show that every protocol achieving asymptotically faithful compression must use a  quantum memory.
Optimality of the protocols is proven later in Section \ref{sec-optimality}.
Finally,   the conclusions are drawn in Section \ref{sec-discussion}.

\section{Main result.}\label{sec-result}
The main result of this work is the optimal compression of identically prepared quantum states. We consider two categories of states: generic finite dimensional (i.e. qudit) states and infinite-dimensional displaced thermal states.
%We search for the minimal amount of memory needed to store the states so that they can be retrieved  with asymptotically vanishing  error.

%The memory cost  depends on the family from which the states are drawn.
% For example, if a family of states is contained in another, the former will require less memory than the latter.
%It is then crucial to specify which family of quantum states we are considering.
 Let us start from the first category.    For a quantum system of dimension $d<\infty$,  also known as qudit,  we consider {\em generic states} described by density matrices with  full rank and non-degenerate spectrum.
     We parametrize the states of a $d$-dimensional quantum system as
\begin{align}\label{state}
\rho_{\theta}=U_\xi \, \rho_0(\mu) \, U^\dag_\xi  \, ,\quad \theta  =  (  \xi,\mu) \quad\xi\in\R^{d(d-1)}\quad\mu  \in\R^{d-1}
\end{align}
where  $\rho_0(\mu)$ is the fixed  state
\begin{align}\label{rhomu}
\rho_0 (\mu)  =     \sum_{j=1}^{d}  \,  \mu_j \,   |j\>\<j|\, \qquad\mu_d:=1-\sum_{k=1}^{d-1}\mu_k,
\end{align}
with spectrum ordered as  $\mu_1>\cdots>\mu_{d-1}>\mu_d>0$, while $U_\xi$ is the unitary matrix defined by
\begin{align}\label{parameters}
U_\xi&=\exp\left[i\left(\sum_{1\le j<k\le d}\frac{\xi^{\rm I}_{j,k}T^{\rm I}_{j,k}+\xi^{\rm R}_{j,k}T^{\rm R}_{k,j}}{\sqrt{\mu_j-\mu_k}}\right)\right]
\end{align}
 Here $\xi$ is a vector of real parameters $ (\xi^{\rm R}_{j,k},\xi^{\rm I}_{j,k})_{1\le j<k\le d}$,  and  $T^{\rm I}$   ($T^{\rm R}$)
 is the matrix  defined by $T^{\rm I}_{j,k}:=iE_{j,k}-iE_{k,j}$   ($T^{\rm R}_{k,j}:=E_{j,k}+E_{k,j})$, where
$E_{j,k}$   is the $d\times d$ matrix with  1 in the entry $(j,k)$ and 0 in all the other entries.

We consider $n$-copy qudit state families, denoted   as $\{\rho_\theta^{\otimes n}\}_{\theta\in\set{\Theta}}$ where $\set{\Theta}$ is the set of possible vectors $\theta=  (\xi,\mu)$.
 We call the components of the vector $\xi$  {\em quantum parameters} and the components of the vector  $\mu$  {\em classical parameters}.  The classical parameters determine the eigenvalues of the density matrix  $\rho_\theta$, while the quantum parameters determine the eigenbasis.

 We say that a parameter is {\em independent} if it can vary continuously while  the other parameters are kept  fixed.   For a given family of states,   we denote by $f_c$ ($f_q$) the maximum number of independent classical (quantum) parameters describing states in the family.  For example, the family of all diagonal density matrices $\rho_0(\mu)$ in Eq. (\ref{rhomu}) has $d-1$ independent parameters.
 %, corresponding to the set $\set F_{\max}=  \{\mu_1,\mu_2,\dots,  \mu_{d-1}\}$.
 %For a general family, the maximal set of independent classical parameters must be  a subset of $\set F_{\rm max}$.
 The family of all quantum states in dimension $d$ has $d^2-1$ independent  parameters, of which $d-1$ are classical and   $d(d-1)$ are quantum.
  In general, we will assume that the family $\{\rho_\theta^{\otimes n}\}_{\theta \in\set \Theta}$ is such that every component of the vector $\theta$ is either independent or fixed to a specific value.

%We consider the compression of an arbitrary family of $n$ identically prepared non-degenerate  qudit states. Note that, by ``non-degenerate'', we mean that the eigenvalues $\mu_1,\dots,\mu_{d-1}$ are distinct and strictly greater than zero.
Let us  introduce now the second category of states that are relevant in this paper: the displaced thermal states \cite{vourdas1986superposition,marian1993squeezed}.  Displaced thermal states are a type of infinite-dimensional states frequently encountered in quantum optics \cite{saleh2013photoelectron}.
\iffalse  Physically, a displaced thermal state can  be regarded as the result of noise on a coherent state
\begin{align*}
|\alpha\>=D_\alpha|0\>=e^{-\frac{|\alpha|^2}2}\sum_{k=0}^{\infty}\frac{\alpha^k}{\sqrt{k!}}|k\>\qquad\alpha=|\alpha|e^{i\varphi}\quad \varphi\in[0,2\pi) \, ,
\end{align*}
\fi
%A displaced thermal state can be regarded as the state of a laser under noisy conditions.
%The displacement $D_\alpha$ is now applied to a mixed state, instead of to the vacuum state $|0\>$, generating the displaced thermal state as:
Mathematically, they have the form
\begin{align}
\rho_{\alpha,\beta}=D_{\alpha}\,\rho^{{\rm thm}}_{\beta}D^\dag_{\alpha}
\end{align}
where $D_\alpha=\exp(\alpha \hat{a}^\dag-\bar{\alpha} \hat{a})$ is the {\em displacement operator}, defined in terms of a complex parameter $\alpha \in \C$ (the {\em displacement}),  and  $\hat{a}$ is the  annihilation operator,  satisfying the relation $[\hat a, \hat a^\dag] =1$, while  $\rho^{{\rm thm}}_{\beta}$ is a {\em thermal state},  defined as
\begin{align}
\rho^{{\rm thm}}_{\beta}&:=(1-\beta)\sum_{j=0}^{\infty}\beta^j\,  |j\>\<j| \, ,
\end{align}
where $\beta \in  [0,1) $ is a real parameter, here called the {\em thermal parameter}, and the basis $\{  |j\> \}_{j \in \N}$ consists of the eigenvectors of $\hat a^\dag \hat a$.  For  $\beta =0$, the the displaced thermal states are pure.   Specifically, the state $\rho_{\alpha,  \beta=  0 }$ is the projector on the coherent state \cite{glauber1963coherent}  $|\alpha\> :  =  D_\alpha\,  |0\>$.
%\begin{align*}
%|\alpha\>=D_\alpha|0\>=e^{-\frac{|\alpha|^2}2}\sum_{j=0}^{\infty}\frac{\alpha^j}{\sqrt{j!}}|j\> \, .
%\end{align*}

In the context of quantum optics, infinite dimensional systems are often called {\em modes}.  We will consider  the compression of  $n$  modes, each prepared in the same displaced thermal state.  We denote the $n$-mode states as   $\{\rho_{\alpha,\beta}^{\otimes n}\}_{(\alpha,\beta)\in\set{\Theta}}$ with $\set{\Theta}=\set{\Theta}_{\alpha}\times\set{\Theta}_{\beta}$ being the parameter space. There are three real parameters for the displaced thermal state family: the thermal parameter $\beta$, the amount  of displacement $|\alpha|$, and the phase $\varphi  = \arg \alpha$.   Here, $\beta$ is a classical parameter, specifying the eigenvalues, while $|\alpha|$ and $\varphi$ are quantum parameters, determining the eigenstates.
   We will assume that each of the three parameters  $\beta, |\alpha|$ and $\varphi$  is either independent, or fixed to a determinate value.

A compression protocol consists of two components: the encoder, which compresses the input state into a memory, and the decoder, which recovers the state from the memory. The compression protocol  for $n$ identically prepared systems is  represented by a couple of quantum channels (completely positive trace-preserving linear maps) $(\map{E}_{n}, \map{D}_{n})$ characterizing the encoder and the decoder, respectively.
We focus on {\em asymptotically faithful} compression protocols, whose error vanishes in the large $n$ limit. As a measure of error, we choose  the supremum of the trace distance between the original state and the recovered state $\map{D}_{n}\circ\map{E}_{n}(\rho_\theta^{\otimes n})$
\begin{align}
\epsilon:=\sup_{\theta\in\set{\Theta}}\frac12  \Big \|\rho_\theta^{\otimes n}-\map{D}_{n}\circ\map{E}_{n}(\rho_\theta^{\otimes n})  \Big\|_1 \, . \label{errordefinition}
\end{align}
The main result of the paper is the following:
\begin{mythm}\label{main}
Let $\{\rho_\theta^{\otimes n}\}_{\theta=(\xi,\mu)\in\set{\Theta}}$ be a generic family of $n$-copy qudit states with $f_c$ independent classical parameters and $f_q$ independent quantum parameters.
For any $\delta\in(0,2/9)$, the states in the family can be compressed into $[(1/2+\delta)f_c+(1/2)f_q]\log n$ classical bits and $(f_q\delta)\log n$ qubits with an error $\epsilon=O\left(n^{-\kappa(\delta)}\right)+O\left( n^{-\delta/2}\right)$, where $\kappa(\delta)$ is the error exponent of Q-LAN \cite{kahn-thesis}   (cf.  Eq. (\ref{kappa}) in the following). The protocol is  optimal, in the sense that any compression protocol using a memory of size $[(f_c+f_q)/2-\delta']\log n$ with $\delta'>0$ cannot be  asymptotically faithful.

The same results hold for a family $\{\rho_{\alpha,\beta}^{\otimes n}\}_{(\alpha,\beta)\in\set{\Theta}}$ of displaced thermal states, except that in this case the error is only $\epsilon=O\left( n^{-\delta/2}\right)$.
\end{mythm}
Theorem \ref{main} is a sort of ``equipartition theorem'',  stating that each independent parameter requires   a memory of size $(1/2+\delta)\log n$.  When the parameter is classical, the required memory is fully classical; when the parameter is quantum, a quantum memory of $\delta\log n$ qubits is required.
\iffalse
Theorem \ref{main} provides  the optimal compression protocols for many $n$-copy qudit state families, including:
\begin{enumerate}
\item The full quantum family containing states of the form $\rho_{\theta}^{\otimes n}$ where $\rho_\theta$ is an arbitrary non-degenerate qudit state.  In this case, one has $f_c=d-1$ and $f_q=d(d-1)$.
\item The full classical  family containing diagonal states of the form $\rho_{\mu}^{\otimes n}$ with $\mu_1>\mu_2>\cdots>\mu_d>0$.   In this case, one has $f_c=d-1$ and $f_q=0$.
Our result shows that probability distributions from this family can be compressed into $(d-1)/2  \log n$ classical bits, retrieving the result of  \cite{classical}.
\item The phase-covariant qudit family ($f_c=0$ and $f_q=d-1$), containing states of the form $\rho_{\phi}^{\otimes n}$ with $\rho_\phi$ generated by multi-phase evolutions of a fixed non-diagonal state:
\begin{align*}
\rho_{\phi}=U_{\phi}\rho_0 U_{\phi}^\dag\qquad U_{\phi}=\exp\left(\sum_{k=1}^{d-1}\phi_k |k\>\<k|\right)
\end{align*}
where $\{|k\>\}$ is an orthonormal basis. States from this family are frequently considered in quantum state estimation \cite{multiphase1,multiphase2}.
\end{enumerate}
\fi

\section{Compression of displaced thermal states}\label{sec-thermal}

In this section, we focus on  the compression of identically prepared displaced thermal states.
  We separately treat  eight possible cases,  corresponding to the possible combinations where  the three parameter $\beta$, $|\alpha|$ and $\varphi$ are either independent or fixed.   The total memory cost for each case is determined by Theorem \ref{main} and is  summarized in Table \ref{table}.   On the other hand, the errors for all  cases satisfy the unified bound
\begin{align}
\epsilon=O\left(n^{-\delta/2}\right).
\end{align}

\begin{table*}[!t]
\centering
\caption{\label{table} Compression rate for different state families. Here $\delta>0$ is an arbitrary positive constant.}
\begin{tabular}{lcccccccc}
Case &\quad & displacement $\alpha=|\alpha|e^{i\varphi}$ &\quad & thermal parameter $\beta$ & \quad & quantum bits &\quad &classical bits\\
\hline
 0 & & fixed & & fixed & & 0 & & 0\\
 1 & & fixed & & independent & &  0 & & $(1/2+\delta)\log n$\\
  2 & & independent & & fixed & & $2\delta\log n$ & & $\log n$\\
3 &  & independent & & independent & &$2\delta\log n$ & &$(3/2+\delta)\log n$ \\
4 & & $\varphi$ independent; $|\alpha|$ fixed & & fixed & & $\delta\log n$ & & $(1/2)\log n$ \\
5 & & $\varphi$ independent; $|\alpha|$ fixed & & independent & & $\delta\log n$ & & $(1+\delta)\log n$ \\
 6 & & $|\alpha|$ independent; $\varphi$ fixed & & fixed & & $\delta\log n$ & & $(1/2)\log n$\\
 7 & & $|\alpha|$ independent; $\varphi$ fixed & & independent & & $\delta\log n$ & & $(1+\delta)\log n$ \\

\hline
\end{tabular}
\end{table*}

\subsection{Quantum optical techniques used in  the compression protocols}
To construct compression protocols for the displaced thermal states, we adopt several tools in quantum optics. As a preparation, we introduce three quantum optical tools which are key components of the compression protocols: the beam splitter, the heterodyne measurement, and the quantum amplifier.
\begin{itemize}
\item{ \em Beam splitter.} A beam splitter is a linear optical device implementing the  unitary gate
    \begin{align}
    U_{\tau} = \exp\left[ i \tau (\hat{a}_1^\dag \hat{a}_2 + \hat{a}_1 \hat{a}_2^\dag)\right],
    \end{align}
 where   $\tau $ is a real parameter,
 %{\color{blue}whose cosine value is called the {\em transmittance}},
 and $\hat{a}_1$ and $\hat{a}_2$ are the annihilation operators associated to the two systems.

Beam splitters can be used to split or merge laser beams. For example, one can use a  beam splitter with $\tau=\pi/4$ to merge two identical coherent states into a single coherent state with larger amplitude:
    \begin{align}\label{yuanshiBS}
    U_{\pi/4} |\alpha\>_1 \otimes |\alpha\>_2 = |\sqrt2 \alpha\>_1 \otimes |0\>_2.
    \end{align}
    We will use beam splitters to manipulate the information about the parameter $\alpha$ in the  $n$-copy state $\rho_{\alpha,\beta}^{\otimes n}$.     In particular,  we  will  make frequent use of the beam splitter unitary  $U_\tau$ that implements the transformation
     \begin{align}\label{bizuBS}
    U_{\tau}\left(\rho_{\alpha_0,\beta}\otimes\rho_{\alpha_1,\beta}\right)U_{\tau}^\dag=\rho_{\alpha_2,\beta}\otimes\rho^{\rm thm}_{\beta}
    \end{align}
    where the displaced thermal states satisfy the relation $\alpha_0=|\alpha_0|e^{i\varphi}$, $\alpha_1=|\alpha_1|e^{i\varphi}$, and $\alpha_2=\sqrt{|\alpha_0|^2+|\alpha_1|^2}e^{i\varphi}$.
%Beam splitters also frequently appear in groups, implementing unitary gates on multiple optical modes.

\item{ \em Heterodyne measurement.} The heterodyne measurement (see e.g. Section 3.5.2 of \cite{busch1997operational}) is a common measurement in continuous variable quantum optics.  Here, the measurement outcome is a complex number $\hat \alpha$ and the corresponding  measurement operator is the projector on the coherent state $|\hat \alpha\>$.
The heterodyne POVM $\{ \frac{\d^2 \hat\alpha}{\pi} \,|\hat\alpha\>\<\hat\alpha| \}$  is normalized in such a way that the integral over the complex plane gives the identity operator, namely
 \begin{align}\int \frac{\d^2 \hat\alpha}{\pi} \ |\hat\alpha\>\<\hat\alpha| = I \, .
 \end{align}

    We will use heterodyne measurements to estimate the amount of displacement of a displaced thermal state. For a displaced thermal state $\rho_{\alpha,\beta}$, the conditional probability density of finding the outcome $\hat \alpha$ is
     \begin{align}
    Q(\hat{\alpha}|\alpha,\beta)=&~ \frac{1}{\pi}\<\hat\alpha|\rho_{\alpha,\beta}|\hat\alpha\> \nonumber\\
    =&~ \frac{1}{\pi}\<\hat\alpha-\alpha|\rho_\beta^{\rm thm}|\hat\alpha-\alpha\> \nonumber\\
    =&~ \frac{(1-\beta)}{\pi}\exp[-(1-\beta)|\hat{\alpha}-\alpha|^2]. \label{heterodyne}
    \end{align}
\item{ \em Quantum amplifier.} A quantum amplifier is a device that  increases the intensity of quantum light while preserving its phase information, namely a device which approximately implements the process $|\alpha\>\to |\gamma \alpha\>$ for $\gamma>1$. Quantum amplification is an analogue of approximate quantum cloning for finite-dimensional systems, since coherent or displaced thermal states can be merged and split in a reversible fashion [cf. Eqs. (\ref{yuanshiBS}), (\ref{bizuBS})]. In this work, we use the following amplifier \cite{amplifier}:
\begin{align}
\nonumber
\map{A}^{\gamma}(\rho)&=\Tr_B\Big[e^{\cosh^{-1}(\sqrt{\gamma}) (\hat{a}^\dag \hat{b}^\dag-\hat{a}\hat{b})}(\rho\otimes |0\>\<0|_B)   \\
  &  
 \qquad   \qquad \times e^{\cosh^{-1}(\sqrt{\gamma})(\hat{a}\hat{b}-\hat{a}^\dag \hat{b}^\dag)}\Big]  
  \label{amplifier}
\end{align}
where $\hat{a}$ and $\hat{b}$ are the annihilation operators of the input mode and the ancillary mode $B$, and $\gamma$ is the amplification factor.
\end{itemize}

We now show details of the compression protocol for each case.

\subsection{Case 1: fixed $\alpha$, independent $\beta$}

Let us  start from Case 1, where the thermal parameter  $\beta$ is the only  independent parameter. Note that the input state can be regarded as the state of $n$ optical modes with each mode in a displaced thermal state, and thus the compression protocol can be regarded as a sequence of operations on the $n$-mode system.

Since $\alpha$ is known, we can get rid of the displacement using a certain unitary and convert the displaced thermal states into (undisplaced) thermal states.
For the compression of thermal states, we have the following lemma:
\begin{Lemma}[Compression of identically prepared thermal states]\label{lemma-thermal}
Let $\{(\rho^{{\rm thm}}_\beta)^{\otimes n} \, \}_{\beta\in  [ \beta_{\min}, \beta_{\max}]}$ be a family of $n$-copy thermal states.  For any $\delta>0$, there exists a protocol $\left(\map{E}^{{\rm thm}}_{n,\delta},\map{D}^{{\rm thm}}_{n,\delta}\right)$ that compresses $n$ copies of a thermal state $\rho^{{\rm thm}}_\beta$ into $(1/2+\delta)\log n$ classical bits with error
\begin{align}\label{error-thermal}
\epsilon_{\rm thm}=O\left(n^{-\delta}\right).
\end{align}
\end{Lemma}
The proof of the above lemma can be found in the appendix. Note  that no quantum memory is required to encode thermal states, in agreement with the intuition that $\beta$ is classical, because all the states $\rho^{{\rm thm}}_\beta$ are diagonal in the same basis.

For any $\delta>0$, the compression protocol for Case 1 (fixed $\alpha$, independent $\beta$) is constructed as follows:
\begin{itemize}
\item \emph{Encoder.}
\begin{enumerate}
\item Transform  each input copy with the displacement operation $\map{D}_{-\alpha}$, defined by  $\map{D}_{-\alpha}(\cdot):={D}_{-\alpha}\cdot {D}_{-\alpha}^\dag$,
where $D_{-\alpha}=\exp(-\alpha \hat{a}^\dag+\bar{\alpha} \hat{a})$ is the displacement operator. The displacement operation transforms each input copy $\rho_{\alpha,\beta}$ into the thermal state $\rho^{{\rm thm}}_\beta$.
%\item The resultant state has the form $(\rho^{\rm thm}_\beta)^{\otimes n}$.
\item Apply the thermal state encoder $\map{E}^{{\rm thm}}_{n,\delta}$ in Lemma \ref{lemma-thermal} on the $n$-mode state and the outcome is encoded in a classical memory.
\end{enumerate}
\item \emph{Decoder.}
\begin{enumerate}
\item Use the thermal state decoder $\map{D}^{{\rm thm}}_{n,\delta}$ in Lemma \ref{lemma-thermal} to recover the $n$ copies of the thermal state $\rho^{{\rm thm}}_\beta$ from the classical memory. \item Perform the displacement operation $\map{D}_{\alpha}$ on each mode.
\end{enumerate}
\end{itemize}
Obviously, the memory cost and the error of the above protocol are given by Lemma \ref{lemma-thermal}.

\subsection{Case 2: fixed $\beta$, independent $\alpha$}
Next we study the case when the displacement $\alpha$ is independent, while the thermal parameter $\beta$ is fixed.
%Let us first look at the idea  before going into the details of the compression protocol.
The heuristic idea of the compression protocol is to gently test the input state, in order to extract information about the parameter $\alpha$.
 % One could consider  To save as much quantum memory as possible, we may consider to estimate {\color{red} estimating} $\alpha$ and store  {\color{red} storing} the outcome in a classical memory. However, any estimate of $\alpha$ comes at the price of disturbing the input state, and, as shown later in Section \ref{sec:classical}, the  distortion   {\color{red} disturbance} caused by measurements on all input copies is so large that the state cannot be recovered faithfully. Instead of measuring all the $n$ copies,
  The ``gentle test" is based on a heterodyne measurement, performed  on a small fraction of the $n$ input copies.   The information gained by the measurement is then used to perform suitable encoding operations on the remaining copies.
% Such an idea of

\begin{figure}  [b!]
\begin{center}
  \includegraphics[width=1\linewidth]{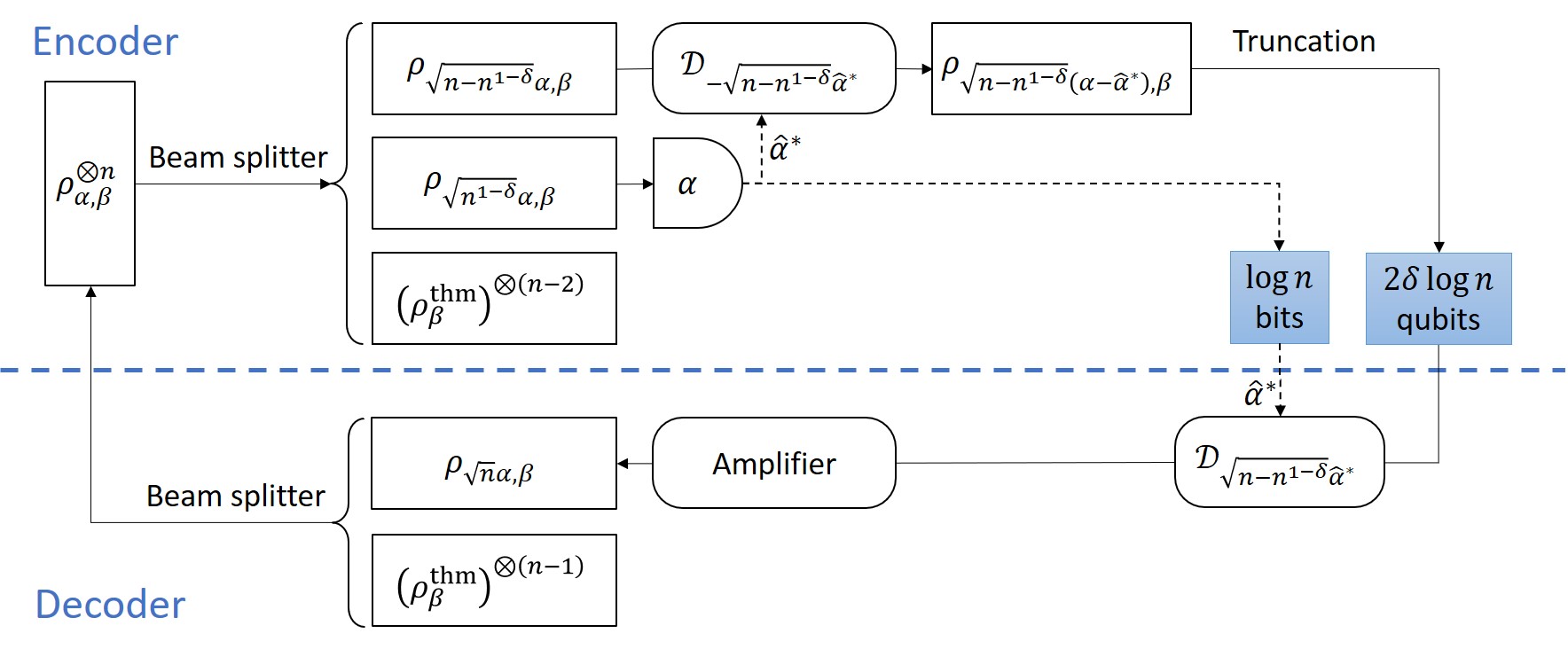}
  \end{center}
\caption{\label{fig-case2}
  {\bf Compression protocol for   displaced thermal states with fixed $\beta$  and independent $\alpha$. }  }
\end{figure}

 Let us see the details of the compression protocol. The key observation is that  $n$ identically displaced thermal states  are unitarily equivalent to a  single displaced thermal state, with the displacement scaled up by a factor $\sqrt n$, times the   product of   $n-1$ undisplaced thermal states. In formula, one has \cite{BS,BS2}
\begin{align}\label{bs}
\rho_{\alpha,\beta}^{\otimes n} =  U_{\rm BS}^\dag\,     \Big (  \rho_{\sqrt{n}\alpha,\beta}\otimes (\rho^{{\rm thm}}_{\beta})^{\otimes (n-1)}\Big)  \,  U_{\rm BS}\, \, ,
\end{align}
where $U_{\rm BS}$ is  a suitable unitary gate,  realizable with a circuit of  beam splitter gates  as in Eq.  (\ref{bizuBS}).
Using  Eq. (\ref{bs}), we can construct a protocol that  separately processes  the displaced thermal state $\rho_{\sqrt n \alpha, \beta}$ and the $n-1$ thermal modes, up to a gentle testing of the input state.

For any $\delta>0$, the protocol for Case 2 (fixed $\beta$, independent $\alpha$) runs as follows (see also Fig \ref{fig-case2} for a flowchart illustration):
\begin{itemize}
\item \emph{Preprocessing.} A preprocessing procedure is needed in order to store the estimate of $\alpha$: Divide the range of $\alpha$ into $n$ intervals, each labeled by a point $\hat{\alpha}_i$ in it,
%{\color{red}  the notation is confusing because later you will use the symbol $\hat \alpha$ for the estimate}
 so that $|\alpha'-\alpha''|=O(n^{-1/2})$ (note that $\alpha$ is complex) for any $\alpha',\alpha''$ in the same interval.
\item \emph{Encoder.}
\begin{enumerate}
\item Perform the unitary channel
$\map{U}_{\rm BS}(\cdot)=U_{\rm BS}\,\cdot\,U_{\rm BS}^\dag$
on the input state, where $U_{\rm BS}$ is the unitary defined by Eq. (\ref{bs}).
The output state has the form $\rho_{\sqrt{n}\alpha,\beta}\otimes (\rho^{{\rm thm}}_{\beta})^{\otimes (n-1)}$.
\item Send the first and the last mode through a group of beam splitters (\ref{bizuBS}) that implements the transformation $\rho_{\alpha,\beta}\otimes\rho_{\beta}^{\rm thm}\to\rho_{\sqrt{n-n^{1-\delta}}\alpha,\beta}\otimes\rho_{\sqrt{n^{1-\delta}}\alpha,\beta}$. The $n$-mode state is now $\rho_{\sqrt{n-n^{1-\delta}}\alpha,\beta}\otimes (\rho^{{\rm thm}}_{\beta})^{\otimes (n-2)}\otimes\rho_{\sqrt{n^{1-\delta}}\alpha,\beta}$.
\item Estimate $\alpha$ by performing   the heterodyne measurement $\{\frac{\d^2\alpha'}{\pi}|\alpha'\>\<\alpha'|\}$ on the last mode.   If the measurement outcome is $\alpha'$,  use the value   $\hat{\alpha}=\alpha'/\sqrt{n^{1-\delta}}$ as an estimate for the displacement $\alpha$.
The conditional probability distribution of the estimate is $Q(\hat{\alpha}|\alpha,\beta)=\frac{(1-\beta)}{\pi}\exp[-(1-\beta)n^{1-\delta}|\hat{\alpha}-\alpha|^2]$ [
cf. Eq. (\ref{heterodyne})].
\item Encode the label $\hat{\alpha}^\ast$ of the interval  containing $\hat{\alpha}$ in a classical memory.
\item Displace the first mode with $\map{D}_{-\sqrt{n-n^{1-\delta}}\hat{\alpha}^\ast}$.
\item Truncate the state of the first mode in the photon number basis. The truncation is described by the channel $\map{P}_{n^{2\delta}}$:
\begin{align}\label{P-em}
\map{P}_{n^{2\delta}}(\rho)={P}_{n^{2\delta}}\rho {P}_{n^{2\delta}}+(1-\Tr[{P}_{n^{2\delta}}\rho])|0\>\<0|
\end{align}
where
\begin{align}
{P}_{n^{2\delta}}&=\sum_{m=0}^{n^{2\delta}}|m\>\<m|\label{p-channel-unknown}.
\end{align}
The output state on the first mode is encoded in a quantum memory.
\end{enumerate}

\item \emph{Decoder.}
\begin{enumerate}
\item Read $\hat{\alpha}^\ast$ and perform the displacement operation $\map{D}_{\sqrt{n-n^{1-\delta}}\hat{\alpha}^\ast}$ on the state of the quantum memory.
\item  Apply a quantum amplifier  $\map{A}^{\gamma_n}$ (\ref{amplifier}) with 
\begin{align}\label{gamma-n}
\gamma_n=\frac{1}{1-n^{-\delta}}
\end{align}
  to the output. %state so as to compensate the loss of the {\color{yellow} $n^\delta$ modes } {\color{blue} amount of displacement} in measurement.
\item Prepare  $(n-1)$ modes in the thermal state $\rho^{{\rm thm}}_\beta$, and perform on all the $n$ modes the unitary channel $\map{U}^{-1}_{\rm BS}$:
\begin{align}
\map{U}_{\rm BS}^{-1}(\rho)=U^\dag_{\rm BS}\,\rho\, U_{\rm BS}.
\end{align}
\end{enumerate}
\end{itemize}
%(note that only the leading order matters since $\delta$ can be made arbitrarily small)
The total memory cost consists of two parts: $\log n$ bits for encoding the (rounded) value $\hat{\alpha}^\ast$ of the estimate and $2\delta\log n$ qubits for encoding the first mode (in a displaced thermal state).
%Overall, the protocol requires $2\delta\log n$ qubits and $\log n$ classical bits.

%On the other hand, the error of the protocol can be split into three terms as {\color{red} add explanations to the error terms. ref for Q: Eq. (\ref{heterodyne}). ref for definition of error Eq. (\ref{errordefinition})}

Let us analyze the error of the protocol. To upper bound the error, we first note that, with high probability, our estimate $\hat\alpha$ is close to the correct value, say $|\hat \alpha  -  \alpha| \le  f(n)$ for some function $f$ vanishing for large $n$.     When  this happens,  we can bound the error introduced by the truncation $\map{P}_{n^{2\delta}}$ and by the amplification $\map{A}^{\gamma_n}$.  Otherwise,  we just use the trivial error bound $ \|\rho_{\alpha,\beta}^{\otimes n}-\map{D}_{n}\circ\map{E}_{n}(\rho_{\alpha,\beta}^{\otimes n})  \Big\|_1  \le 2$.
  In this way, we obtain the bound
  \begin{align}
 \nonumber  \epsilon &= \sup_{\alpha,\beta}\frac12  \Big \|\rho_{\alpha,\beta}^{\otimes n}-\map{D}_{n}\circ\map{E}_{n}\left(\rho_{\alpha,\beta}^{\otimes n}\right)  \Big\|_1 \, \\
\nonumber &\le \frac12    \sup_{\alpha,\beta}  \Big\{ \sup_{\hat{\alpha}^\ast:|\hat{\alpha}^\ast-\alpha|\le f(n)}    \\
\nonumber & \qquad  \quad  \Big\|\map{A}^{\gamma_n} \circ     \map{D}_{\sqrt{n-n^{1-\delta}}\hat{\alpha}^\ast} \circ   \map{P}_{n^{2\delta}}    \left(\rho_{\sqrt{n-n^{1-\delta}}(\alpha-\hat{\alpha}^\ast),\beta}\right) \\
& \qquad \quad -\rho_{\sqrt{n}\alpha,\beta}\Big\|_1\Big\}  +  {P}    (\alpha,\beta, n) \, ,
 \label{equazione}
 \end{align}
 where  $P  (\alpha,\beta,n)$ is the probability that $\hat \alpha$ deviates from $\alpha$ by more than $f(n)$, given by
   \begin{align}
\nonumber    &  P  (\alpha,\beta,n)    =   \int_{|\hat{\alpha}-\alpha|>f(n)} \,  \d^2 \hat \alpha  \,   Q( \hat{\alpha}|\alpha,\beta) \\
  &  =     \int_{|\hat{\alpha}-\alpha|>f(n)} \,  \frac{\d^2 \hat \alpha}{\pi}  \,      (1-\beta)\exp[-(1-\beta)n^{1-\delta}|\hat{\alpha}-\alpha|^2]  \, ,
   \end{align}
having used Eq. (\ref{heterodyne}) in the second equality.
    At this point, it is convenient to set $f(n)  = n^{-1/2+3\delta/4}$, so that we obtain the relation
   \begin{align}
\nonumber  &P  (\alpha,  \beta, n)\\
\nonumber &=\int_{|\hat{\alpha}-\alpha|>n^{-1/2+3\delta/4}} \,   \,  \frac{\d^2 \hat \alpha}{\pi}  \,(1-\beta)\exp[-(1-\beta)n^{1-\delta}|\hat{\alpha}-\alpha|^2]\\
& =e^{-\Omega(n^{\delta/2})}.
\end{align}
Inserting this relation in Eq. (\ref{equazione}), we obtain the bound
\begin{align}
\nonumber   &\epsilon  \le \frac12    \sup_{\alpha,\beta}  \Big\{ \sup_{\hat{\alpha}^\ast:|\hat{\alpha}^\ast-\alpha|\le f(n)}  \\
\nonumber   &  \qquad  \Big\|\map{A}^{\gamma_n}   \circ \map{D}_{\sqrt{n-n^{1-\delta}}\hat{\alpha}^\ast}  \circ\map{P}_{n^{2\delta}}    \left(\rho_{\sqrt{n-n^{1-\delta}}(\alpha-\hat{\alpha}^\ast),\beta}\right)\\
  \label{altraequazione}
   &\qquad -\rho_{\sqrt{n}\alpha,\beta}\Big\|_1\Big\}  +  e^{-\Omega(n^{\delta/2})}  \, .
\end{align}
Now, we have to bound the first term in the right hand side.
To this purpose, we split it into two terms, as  follows
\begin{align}
   \nonumber    &\left\|\map{A}^{\gamma_n}   \circ\map{D}_{\sqrt{n-n^{1-\delta}}\hat{\alpha}^\ast}  \circ\map{P}_{n^{2\delta}}    \left(\rho_{\sqrt{n-n^{1-\delta}}(\alpha-\hat{\alpha}^\ast),\beta}\right)  -\rho_{\sqrt{n}\alpha,\beta}\right\|_1      \\
    \nonumber
   & \quad \le
     \Big\|\map{A}^{\gamma_n}   \circ \map{D}_{\sqrt{n-n^{1-\delta}}\hat{\alpha}^\ast} \circ \map{P}_{n^{2\delta}}    \left(\rho_{\sqrt{n-n^{1-\delta}}(\alpha-\hat{\alpha}^\ast),\beta}\right) \\
  \nonumber   &\qquad  \quad  -  \map{A}^{\gamma_n}    \circ\map{D}_{\sqrt{n-n^{1-\delta}}\hat{\alpha}^\ast}      \left(  \rho_{\sqrt{n-n^{1-\delta}}(\alpha-\hat{\alpha}^\ast),\beta}\right)   \Big\|_1       \\
   \nonumber & \qquad   \quad + \left\|  \map{A}^{\gamma_n} \circ   \map{D}_{\sqrt{n-n^{1-\delta}}\hat{\alpha}^\ast}   \left(   \rho_{\sqrt{n-n^{1-\delta}}(\alpha-\hat{\alpha}^\ast),\beta} \right)  -\rho_{\sqrt{n}\alpha,\beta}\right\|_1    \\
     \nonumber &\le          \Big\|    \map{P}_{n^{2\delta}}    \left(\rho_{\sqrt{n-n^{1-\delta}}(\alpha-\hat{\alpha}^\ast),\beta}\right) -     \rho_{\sqrt{n-n^{1-\delta}}(\alpha-\hat{\alpha}^\ast),\beta}  \Big\|_1  \\
     &  \quad  + \left\|  \map{A}^{\gamma_n}     \left(   \rho_{\sqrt{n-n^{1-\delta}}\alpha ,\beta}   \right) -\rho_{\sqrt{n}\alpha,\beta}\right\|_1       \, .  \label{twomoreterms}
\end{align}
The two terms can be upper bounded individually.   For the first term, we use the relations
 \begin{align}\label{amp-output}
&\map{A}^{\gamma}\left(\rho_{\alpha,\beta}\right)=\rho_{\sqrt{\gamma}\alpha,\beta'}\qquad \beta'=\frac{\beta+\gamma-1}{\gamma}
\end{align}
and
\begin{align}\label{thermal-property}
\left\|\rho^{{\rm thm}}_{\beta'}-\rho^{{\rm thm}}_{\beta}\right\|_1&\le\frac{2|\beta'-\beta|}{(1-\beta')^2}+O(|\beta'-\beta|^2) \, ,
\end{align}
proven in Appendices \ref{app-amp} and \ref{app-trivial}, respectively.

Using these two relations and Eq. (\ref{gamma-n}), we obtain the bound
\begin{align}
\nonumber &\frac12\sup_{\alpha,\beta}\sup_{\hat{\alpha}^\ast:|\hat{\alpha}^\ast-\alpha|\le n^{-1/2+3\delta/4}}\left\|\map{A}^{\gamma_n}\left(\rho_{\sqrt{n-n^{1-\delta}}\alpha,\beta}\right)-\rho_{\sqrt{n}\alpha,\beta}\right\|_1  \\
\label{mancante}&=O(n^{-\delta}) \, .
\end{align}

The first term in the right hand side of Eq. (\ref{twomoreterms}) can be bounded with the following lemma:
\begin{Lemma}[Photon number truncation of displaced thermal states.]\label{lemma-truncation}
Define the channel $\map{P}_{K}$ as
\begin{align}\label{trun-P-alpha}
\map{P}_{K}(\rho)=P_{K}\rho P_{K}+(1-\Tr[P_{K}\rho])|0\>\<0|
\end{align}
where $P_{K}=\sum_{k=0}^K|k\>\<k|$.
When $K=\Omega\left(|\alpha|^{2+x}\right)$, $\map{P}_{K}$ satisfies
\begin{align}\label{error-coherent-known}
\epsilon(\rho_{\alpha,\beta}):=\frac12\left\|\map{P}_{K}\left(\rho_{\alpha,\beta}\right)-\rho_{\alpha,\beta}\right\|_1=\beta^{\Omega(K^{x/8})}+e^{-\Omega(K^{x/4})}
\end{align}
for any $0\le\beta<1$.
\end{Lemma}
See Appendix \ref{proof-lemma-truncation} for the proof.

In our case,  we are using the projector $P_{n^{2\delta}}$ in Eq. (\ref{p-channel-unknown}), and the displacement is $\sqrt{n-n^{1-\delta}}(\alpha-\hat{\alpha}^\ast)$. Since $\sqrt{n-n^{1-\delta}}|\alpha-\hat{\alpha}^\ast|=O\left( n^{3\delta/4}\right)$, by Lemma \ref{lemma-truncation} we obtain the bound
\begin{align}
\nonumber   &\frac12\sup_{\alpha,\beta}\sup_{\hat{\alpha}^\ast:|\hat{\alpha}^\ast-\alpha|\le n^{-1/2+3\delta/4}}
\Big\|\map{P}_{n^{2\delta}}(\rho_{\sqrt{n-n^{1-\delta}}(\alpha-\hat{\alpha}^\ast),\beta})  \\
  \nonumber & \qquad \qquad\qquad  \qquad \qquad \qquad  -\rho_{\sqrt{n-n^{1-\delta}}(\alpha-\hat{\alpha}^\ast),\beta}\Big\|_1\\
\label{basta}
  &=\beta^{\Omega\left(n^{\delta/12}\right)}+e^{-\Omega(n^{\delta/6})}  \, .
\end{align}

Combining Eqs.   (\ref{altraequazione}), %(\ref{twomoreterms}),
(\ref{mancante}), and (\ref{basta}),  we finally get the error bound
\begin{align}
\nonumber \epsilon&  \le   e^{-\Omega(n^{\delta/2})}  +  O\left(n^{-\delta}\right)   +  \beta^{\Omega\left(n^{\delta/12}\right)}+e^{-\Omega(n^{\delta/6})}     \\
&  =   O\left(n^{-\delta}\right)   .  \label{case2error}
\end{align}

\subsection{Case 3: independent $\alpha$ and $\beta$}\label{subsec:case1}
Case 3 (independent $\alpha$ and $\beta$) can be treated in a similar way as Case 2. The main difference is that, since one mode is consumed in the estimation of $\alpha$, we have to estimate also $\beta$ to reconstruct this mode.
Luckily,  the thermal parameter $\beta$ can be estimated freely (i.e. without disturbing the input state), and thus its estimation strategy is simpler than that of $\alpha$.

\begin{figure}  [b!]
\begin{center}
  \includegraphics[width=1\linewidth]{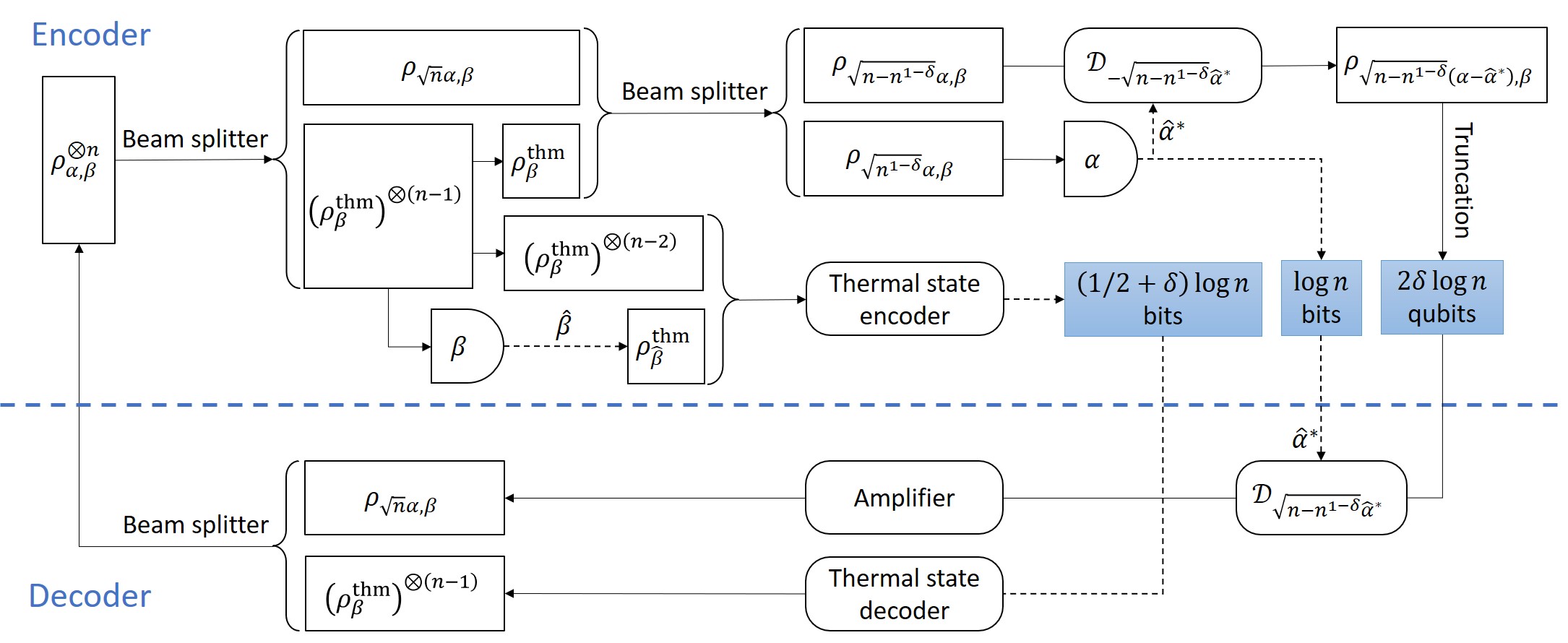}
  \end{center}
\caption{\label{fig-case3}
  {\bf Compression protocol for displaced thermal states with independent $\alpha$ and  $\beta$.}  }
\end{figure}

For any $\delta>0$ we can construct the protocol for Case 3 (independent $\alpha$ and $\beta$) as follows (see also Fig. \ref{fig-case3} for a flowchart illustration):
\begin{itemize}
\item \emph{Preprocessing.} Divide the range of $\alpha$ into $n$ intervals, each labeled by a point $\hat{\alpha}_i$ in it, so that $|\alpha'-\alpha''|=O(n^{-1/2})$  for any $\alpha',\alpha''$ in the same interval.
\item \emph{Encoder.}
\begin{enumerate}
\item
Perform the unitary channel $\map{U}_{\rm BS}(\cdot)=U_{\rm BS}\,\cdot\,U_{\rm BS}^\dag$ on the input state, where $U_{\rm BS}$ is the unitary defined by Eq. (\ref{bs}).
The output state has the form $\rho_{\sqrt{n}\alpha,\beta}\otimes (\rho^{{\rm thm}}_{\beta})^{\otimes (n-1)}$.
\item Estimate $\beta$ with the von Neumann measurement of the photon number on the $n-1$ copies of $\rho^{{\rm thm}}_{\beta}$ and denote by $\hat{\beta}$ the maximum likelihood estimate of $\beta$. Note that  the $n-1$ copies will not be disturbed   by the photon number measurement because they are diagonal in the photon number basis.
\item Send the first and the last mode through a group of beam splitters (\ref{bizuBS}) that implements the transformation $\rho_{\alpha,\beta}\otimes\rho_{\beta}^{\rm thm}\to\rho_{\sqrt{n-n^{1-\delta}}\alpha,\beta}\otimes\rho_{\sqrt{n^{1-\delta}}\alpha,\beta}$. The $n$-mode state is now  $\rho_{\sqrt{n-n^{1-\delta}}\alpha,\beta}\otimes (\rho^{{\rm thm}}_{\beta})^{\otimes (n-2)}\otimes\rho_{\sqrt{n^{1-\delta}}\alpha,\beta}$.
\item Estimate $\alpha$ by performing the heterodyne measurement $\{\frac{\d^2\alpha'}{\pi}|\alpha'\>\<\alpha'|\}$ on the last mode, which yields an estimate $\hat{\alpha}=\alpha'/\sqrt{n^{1-\delta}}$ with the probability distribution $Q(\hat{\alpha}|\alpha, \beta)$ as in Eq. (\ref{heterodyne}). Encode the label $\hat{\alpha}^\ast$ of the interval   containing $\hat{\alpha}$ in a classical memory.
\item Displace the first mode with $\map{D}_{-\sqrt{n-n^{1-\delta}}\hat{\alpha}^\ast}$.
\item Prepare the $n$-th mode in the thermal state $\rho^{{\rm thm}}_{\hat{\beta}}$. The $n$-mode state is now $\rho_{\sqrt{n-n^{1-\delta}}(\alpha-\hat{\alpha}^\ast),\beta}\otimes (\rho^{{\rm thm}}_{\beta})^{\otimes (n-2)}\otimes\rho^{{\rm thm}}_{\hat{\beta}}$.
\item Truncate the state of the first mode, using the channel $\map{P}_{n^{2\delta}}$ defined by Eq. (\ref{trun-P-alpha}).
The output state is encoded in a quantum memory.
\item Use the thermal state encoder $\map{E}^{{\rm thm}}_{n-1,\delta}$ (see Lemma \ref{lemma-thermal}) to compress the remaining $n-1$ modes and encode the output state in a classical memory.
\end{enumerate}
\item \emph{Decoder.}
\begin{enumerate}
\item Read $\hat{\alpha}^\ast$ and perform the displacement $\map{D}_{\sqrt{n-n^{1-\delta}}\hat{\alpha}^\ast}$ on the state of the quantum memory.
\item Apply a quantum amplifier (\ref{amplifier}) $\map{A}^{\gamma_n}$ with $\gamma_n=1/(1-n^{-\delta})$ to the state.
\item Use the thermal state decoder $\map{D}^{{\rm thm}}_{n-1,\delta}$ to recover the other $(n-1)$ modes in the thermal state $\rho^{{\rm thm}}_\beta$ from the memory.
\item Perform the channel $\map{U}^{-1}_{\rm BS}$.
\end{enumerate}
\end{itemize}
%(note that only the leading order matters since $\delta$ can be made arbitrarily small)
The memory cost of the protocol consists of three parts: $\log n$ bits for encoding the (rounded) value $\hat{\alpha}^\ast$ of the estimate, $2\delta\log n$ qubits for encoding the first mode (displaced thermal state), and $(1/2+\delta)\log n$ bits for encoding the other modes (thermal states).
Overall, the protocol requires $2\delta\log n$ qubits and $(3/2+\delta)\log n$ classical bits.

%On the other hand, the error of the protocol can be split into several terms as

On the other hand, the error of the protocol can be analyzed in a similar way as  in Case 2, with the only difference that an extra error is introduced by estimating and compressing the thermal states. The state of the modes after the estimation step is
\begin{align}  \rho^{\rm est}_\beta  =  \left(  \int\d\hat{\beta}\,P( \hat{\beta}|\beta)\,\rho^{{\rm thm}}_{\hat{\beta}}\right)\otimes(\rho^{{\rm thm}}_{\beta})^{\otimes (n-2)}   \, .
 \end{align}
 where $P(\hat{\beta}|\beta)$ is the probability density of estimating $\hat \beta$ when the true value is  $\beta$.
   Applying the thermal state compression to this state, we obtain the output state $\map{D}^{{\rm thm}}_{n,\delta} \circ \map{E}^{{\rm thm}}_{n,\delta}   (  \rho^{\rm est}_\beta)$, whose distance from the initial state can be bounded as
\begin{align}
\nonumber \epsilon_\beta   & :  =   \frac12   \,   \left\|\map{D}^{{\rm thm}}_{n,\delta} \circ \map{E}^{{\rm thm}}_{n,\delta}           (  \rho^{\rm est}_\beta)  - \left( \rho_\beta^{{\rm thm}} \right)^{\otimes (n-1)}  \right\|_1  \\
\nonumber &\le\frac12\sup_{\beta}\Big\{\Big\|  \map{D}^{{\rm thm}}_{n,\delta} \circ \map{E}^{{\rm thm}}_{n,\delta} \left(\rho^{{\rm est}}_{\beta}\right)    \\
\nonumber & \quad   -  \map{D}^{{\rm thm}}_{n,\delta} \circ \map{E}^{{\rm thm}}_{n,\delta}\left[(\rho^{{\rm thm}}_{\beta})^{\otimes (n-1)}\right]    \Big\|_1    \\
 \nonumber  & \quad    +    \left\|\map{D}^{{\rm thm}}_{n,\delta} \circ \map{E}^{{\rm thm}}_{n,\delta}\left[(\rho^{{\rm thm}}_{\beta})^{\otimes (n-1)}\right]-(\rho^{{\rm thm}}_{\beta})^{\otimes (n-1)}\right\|_1\Big\}  \\
&\le\frac12\sup_{\beta}  \left\| \rho^{{\rm est}}_{\beta}     - (\rho^{{\rm thm}}_{\beta})^{\otimes (n-1)}    \right\|_1     +  O\left(n^{-\delta}\right)  \, , \label{epsilonbeta}
\end{align}
having used  Lemma \ref{lemma-thermal} in the last inequality.
The remaining term can be bounded as
\begin{align}\label{ancora}
 \left\| \rho^{{\rm est}}_{\beta}     - (\rho^{{\rm thm}}_{\beta})^{\otimes (n-1)}    \right\|_1 &
 \le  \left\|
  \int\,\d\hat{\beta}P(\hat{\beta}|\beta)\, \Big (  \rho^{{\rm thm}}_{\hat{\beta}}  -\rho^{{\rm thm}}_{\beta}  \Big) \right\|_1  \, .
  \end{align}
Now, we split the integral in the right hand side of Eq. (\ref{ancora}) into two terms, corresponding to the values of $\hat \beta$  in regions  $\set R _\le  :  =  \{  \hat \beta \in \C ~|~   |\hat \beta- \beta|  \le n^{-(1+\delta)/2} )\}$ and $\set R_>  =  \C  \setminus \set R_\le$.    In this way, we obtain the bound 
\begin{align}
\nonumber   & \left\|
  \int\,\d\hat{\beta} P(\hat{\beta}|\beta)\,\rho^{{\rm thm}}_{\hat{\beta}}  -\rho^{{\rm thm}}_{\beta}\right\|_1 \\
  \nonumber  \le& \sup_\beta\sup_{\hat{\beta}\in\set{R}_\le}\left\|\rho^{{\rm thm}}_{\hat{\beta}}-\rho^{{\rm thm}}_{\beta}\right\|_1 + 2\int_{\set{R}_>}\, \d\hat{\beta} P(\hat{\beta}|\beta) \, ,
\end{align}
having used the elementary inequality  $\left\|\rho^{{\rm thm}}_{\hat{\beta}}-\rho^{{\rm thm}}_{\beta}\right\|_1\leq 2$. The first term  in the right hand side is bounded by $O(n^{-(1+\delta)/2})$ using Eq. (\ref{thermal-property}),
while the second error term is bounded by the following property of the maximum likelihood estimate \cite{van-book}
\begin{align}
\int_{|\hat{\beta}-\beta|\ge l/\sqrt{n F_\beta}}\, \d\hat{\beta} P(\hat{\beta}|\beta)\le\erfc\left(\frac{l}{\sqrt{2}}\right)
\end{align}
where $F_\beta=(\beta^2+1)/[\beta(1-\beta)^3]$ is the Fisher information of $\beta$ and $\erfc(x):=(2/\pi)\int_x^{\infty}e^{-s^2}\d s$ is the complementary error function. Picking $l=n^{-\delta/2}\sqrt{F_\beta}$, we have
\begin{align}
\int_{\set{R}_>}\, \d\hat{\beta} P(\hat{\beta}|\beta)\le\erfc\left(\frac{n^{-\delta/2}\sqrt{F_\beta}}{\sqrt{2}}\right) = e^{-\Omega(n^\delta)} \, .
\end{align}

In conclusion, $\epsilon_\beta$ can be bounded as
\begin{align}
\nonumber \epsilon_\beta&\le O(n^{-\delta})+O(n^{-(1+\delta)/2})+e^{-\Omega(n^\delta)}\\
&=O(n^{-\delta}) \, .
\end{align}

The remaining contribution to the error can be bounded as in Eq.  (\ref{case2error}), leading to an overall error of size $O(n^{-\delta})$. 

\subsection{Case 4 (independent $\varphi$, fixed $|\alpha|$ and $\beta$) and Case 6 (independent $|\alpha|$, fixed $\varphi$ and $\beta$).}
In Case 4 (independent $\varphi$, fixed $|\alpha|$ and $\beta$) and Case 6 (independent $|\alpha|$, fixed $\varphi$ and $\beta$), the displacement $\alpha$ is partially known. Such a knowledge allows us to reduce the amount of memory.

The protocols for these two cases 4 and 6 are very similar. Let us start from Case 4, where the phase of the displacement is independent while the modulus is fixed. The protocol for Case 4 (independent $\varphi$, fixed $|\alpha|$ and $\beta$) runs as follows:
\begin{itemize}
\item \emph{Preprocessing.} Divide the range of $\varphi$ into $n^{1/2}$ intervals, each labeled by a point $\hat{\varphi}_i$ in it, so that $|\varphi'-\varphi''|=O(n^{-1/2})$  for any $\varphi',\varphi''$ in the same interval.
\item \emph{Encoder.}
\begin{enumerate}
\item Perform the unitary channel $\map{U}_{\rm BS}(\cdot)$ on the input state to transform it into $\rho_{\sqrt{n}\alpha,\beta}\otimes (\rho^{{\rm thm}}_{\beta})^{\otimes (n-1)}$.
\item Send the first and the last mode through a group of beam splitters (\ref{bizuBS}) that implements the transformation $\rho_{\alpha,\beta}\otimes\rho_{\beta}^{\rm thm}\to\rho_{\sqrt{n-n^{1-\delta/2}}\alpha,\beta}\otimes\rho_{\sqrt{n^{1-\delta/2}}\alpha,\beta}$.
\item Estimate $\varphi$ by the heterodyne measurement $\{\frac{\d^2\alpha'}{\pi}|\alpha'\>\<\alpha'|\}$ on the last mode, which yields an estimate $\hat{\varphi}$ which is the phase of $\alpha'$.
 Encode the label $\hat{\varphi}^\ast$ of the interval  containing $\hat{\varphi}$ in a classical memory.
 \item Displace the first mode with $\map{D}_{-\sqrt{n-n^{1-\delta/2}}\hat{\alpha}^\ast}$ with $\hat{\alpha}^\ast:=|\alpha|e^{i\hat{\varphi}^\ast}$.
 \item Send the state of the first mode through a truncation channel  $\map{P}_{n^{\delta}}$ defined in (\ref{trun-P-alpha}) and encode the output state in a quantum memory.
\end{enumerate}

\item \emph{Decoder.}
\begin{enumerate}
\item Read $\hat{\alpha}^\ast$ and perform the displacement $\map{D}_{\sqrt{n-n^{1-\delta/2}}\hat{\alpha}^\ast}$ on the state of the quantum memory.
\item Apply the quantum amplifier $\map{A}^{\gamma_n}$ with $\gamma_n=1/(1-n^{-\delta/2})$.
\item Prepare the other $(n-1)$ modes in the thermal state $\rho^{{\rm thm}}_\beta$.
\item Perform $\map{U}^{-1}_{\rm BS}$ on the thermal state $\rho^{{\rm thm}}_\beta$ and the quantum memory.
\end{enumerate}
\end{itemize}
The protocol for Case 6 works in the same way except that $|\alpha|$ is estimated instead of $\varphi$.
%(note that only the leading order matters since $\delta$ can be made arbitrarily small)
For both cases the memory cost consists of two parts: $(1/2)\log n$ bits for encoding the (rounded) value $\hat{\alpha}^\ast$ of the estimate and $\delta\log n$ qubits for encoding the first mode (displaced thermal state).
The error can be bounded as previous as
$\epsilon=O\left(n^{-\delta/2}\right)$.

\subsection{Case 5 (fixed $|\alpha|$, independent $\varphi$ and $\beta$) and Case 7 (fixed $\varphi$, independent $|\alpha|$ and $\beta$).}
Case 5 (fixed $|\alpha|$, independent $\varphi$ and $\beta$) and Case 7 (fixed $\varphi$, independent $|\alpha|$ and $\beta$) can be treated in the same way as Case 4 (independent $\varphi$, fixed $|\alpha|$ and $\beta$) and Case 6 (independent $|\alpha|$, fixed $\varphi$ and $\beta$),  except that the thermal parameter $\beta$ is now independent. We illustrate only the protocol for Case 5 (fixed $|\alpha|$, independent $\varphi$ and $\beta$) and the other naturally follows. The protocol runs as follows:
\begin{itemize}
\item \emph{Preprocessing.} Divide the range of $\varphi$ into $n^{1/2}$ intervals, each labeled by a point $\hat{\varphi}_i$ in it, so that $|\varphi'-\varphi''|=O(n^{-1/2})$  for any $\varphi',\varphi''$ in the same interval.
\item \emph{Encoder.}
\begin{enumerate}
\item Perform the unitary channel $\map{U}_{\rm BS}(\cdot)$ on the input state to transform it into $\rho_{\sqrt{n}\alpha,\beta}\otimes (\rho^{{\rm thm}}_{\beta})^{\otimes (n-1)}$.
\item Estimate $\beta$ with the von Neumann measurement of the photon number on the $n-1$ copies of $\rho^{{\rm thm}}_{\beta}$. Denote by $\hat{\beta}$ the maximum likelihood estimate of $\beta$.
\item Send the first and the last mode through a group of beam splitters (\ref{bizuBS}) that implements the transformation $\rho_{\alpha,\beta}\otimes\rho_{\beta}^{\rm thm}\to\rho_{\sqrt{n-n^{1-\delta/2}}\alpha,\beta}\otimes\rho_{\sqrt{n^{1-\delta/2}}\alpha,\beta}$.
\item Estimate $\varphi$ by the heterodyne measurement $\{\frac{\d^2\alpha'}{\pi}|\alpha'\>\<\alpha'|\}$ on the last mode, which yields an estimate $\hat{\varphi}$ which is the phase of $\alpha'$.
 Encode the label $\hat{\varphi}^\ast$ of the interval  containing $\hat{\varphi}$ in a classical memory.
 \item Displace the first mode with $\map{D}_{-\sqrt{n-n^{1-\delta/2}}\hat{\alpha}^\ast}$ with $\hat{\alpha}^\ast:=|\alpha|e^{i\hat{\varphi}^\ast}$.
  \item Prepare the $n$-th mode in the thermal state $\rho^{{\rm thm}}_{\hat{\beta}}$. The $n$-mode state is now $\rho_{\sqrt{n-n^{1-\delta/2}}(\alpha-\hat{\alpha}^\ast),\beta}\otimes (\rho^{{\rm thm}}_{\beta})^{\otimes (n-2)}\otimes\rho^{{\rm thm}}_{\hat{\beta}}$.
\item Send the state of the first mode through a truncation channel  $\map{P}_{n^{\delta}}$ defined in (\ref{trun-P-alpha}) and encode the output state in a quantum memory.
\item Use the thermal state encoder $\map{E}^{{\rm thm}}_{n-1,\delta}$ (see Lemma \ref{lemma-thermal}) to compress the remaining $n-1$ modes and encode the output state in a classical memory.
\end{enumerate}

\item \emph{Decoder.}
\begin{enumerate}
\item Read $\hat{\alpha}^\ast$ and perform the displacement $\map{D}_{\sqrt{n-n^{1-\delta/2}}\hat{\alpha}^\ast}$ on the state of the quantum memory.
\item Apply the quantum amplifier $\map{A}^{\gamma_n}$ with $\gamma_n=1/(1-n^{-\delta/2})$.
\item Use the thermal state decoder $\map{D}^{{\rm thm}}_{n-1,\delta}$ to recover the other $(n-1)$ modes in the thermal state $\rho^{{\rm thm}}_\beta$ from the memory.
\item Perform $\map{U}^{-1}_{\rm BS}$ on the output of $\map{D}^{{\rm thm}}_{n-1,\delta}$ and the quantum memory.
\end{enumerate}
\end{itemize}
%(note that only the leading order matters since $\delta$ can be made arbitrarily small)

The protocol for Case 7 works in the same way except that $|\alpha|$ is estimated instead of $\varphi$.
For both cases the memory cost consists of three parts: $(1/2)\log n$ bits for encoding the (rounded) value $\hat{\alpha}^\ast$ of the estimate, $\delta\log n$ qubits for encoding the first mode (displaced thermal state), and $(1/2+\delta)\log n$ bits for encoding the other modes (thermal states).
Overall, the protocol requires $\delta\log n$ qubits and $(1+2\delta)\log n$ classical bits.
The error can be bounded as previous as $\epsilon=O\left(n^{-\delta/2}\right)$.

\section{Compression of identically prepared finite dimensional systems.}\label{sec-main}
In this section, we study the compression of finite-dimensional non-degenerate quantum systems, using quantum local asymptotic normality and leveraging  on our results on displaced thermal states.   %through quantum local asymptotic normality.
We show that, just as for displaced thermal states, each independent parameter of a qudit family requires $(1/2+\delta)\log n$ memory for any $\delta>0$.
 The compression protocol is introduced in the following.

\subsection{The compression protocol}\label{subsec-protocol}
To construct a compression protocol, we will use the following  techniques:
\begin{itemize}
\item{ \em Quantum local asymptotic normality (Q-LAN).} The quantum version of local asymptotic normality has been derived in several different forms \cite{LAN,LAN2,kahn-thesis,LAN3}. Here we use the version of \cite{kahn-thesis}, which states that $n$ identical copies of a qudit state can be locally approximated by a classical-quantum Gaussian state in the large $n$ limit.

Explicitly, for a fixed point $\theta_0=\left(\xi_0,\mu_0\right)$, one defines the neighborhood
\begin{align}\label{neighborhood}\set{\Theta}_{n,x}(\theta_0)=\left\{\theta=\theta_0+\delta\theta/\sqrt{n},~|~ \|\delta\theta\|_{\infty}\le n^{\frac{x}2}\right\} \,,
\end{align}
where $\|\delta\theta\|_{\infty}$ is the max vector norm  $\|\delta\theta\|_{\infty}:=\max_i (\delta\theta)_i$ and $x\in(0,1)$.
Q-LAN states that every  $n$-fold product state $\rho_{\theta}^{\otimes n}$ with $\theta$ in  the neighborhood  $\set{\Theta}_{n,x}(\theta_0)$ can be approximated by a classical-quantum Gaussian state:
\begin{align}
%\rho_{\theta}^{\otimes n}&\leftrightarrow G_{n,\theta}\qquad\qquad \theta\in\set{\Theta}_{n,x}(\theta_0)\\
\nonumber G_{n,\theta}&=N\left(\delta\mu,V_{\mu_0}\right)\otimes\Phi\left(\delta\xi,\mu_0\right) \\
\Phi\left(\delta\xi,\mu_0\right)&=\bigotimes_{1\le j<k\le d}\rho_{\alpha_{j,k},\beta_{j,k}},\label{LAN-gaussian}
\end{align}
where $(\delta\xi,\delta\mu) = \sqrt{n} (\theta-\theta_0)$, $N\left(\delta\mu,V_{\mu_0}\right)$ is the multivariate normal distribution with mean $\delta\mu$ and covariance matrix $V_{\mu_0}$ (equal to the inverse of the Fisher information of the $(d-1)$-dimensional family of probability distributions $\{\mu\}$, evaluated at $\mu=\mu_0$) and $\rho_{\alpha_{j,k},\beta_{j,k}}$ is the displaced thermal state defined as
\begin{align}\label{LAN-thermal}
\rho_{\alpha_{j,k},\beta_{j,k}}&=D_{\alpha_{j,k}}\rho^{{\rm thm}}_{\beta_{j,k}}D^\dag_{\alpha_{j,k}}\\
\alpha_{j,k}&=\delta\xi^{I}_{j,k}+i\delta\xi^{R}_{j,k}\qquad \beta_{j,k}=\frac{(\mu_{0})_k}{(\mu_0)_j}.
\end{align}
 where $\{(\mu_{0})_j\}_{j=1}^{d-1}$ are components of $\mu_0$ and $(\mu_0)_d:=1-\sum_{k=1}^{d-1}(\mu_0)_k$.
%Recalling that the parameters $(\mu, \xi)$ are contained in the neighbourhood   $\set{\Theta}_{n,x} (\theta_0)$, we obtain that the displacement amplitude is upper bounded as
%\begin{align}\label{max-alpha}
%|\alpha_{j,k}|\le|\alpha_{j,k}|_{\max}:=\frac{n^{\frac{1-x}2}}{2\sqrt{\mu_{0,j}-\mu_{0,k}}}.
%\end{align}
%Here $D_{\alpha}=\exp(\alpha a^\dag-\bar{\alpha} a)$ is the displacement operator and $\rho_{{\rm thm},\beta}=(1-\beta)\sum_i \beta^i|i\>\<i|$ is a thermal state.
In the neighborhood $\set{\Theta}_{n,x}(\theta_0)$ of $\theta_0$, it is possible to construct two quantum channels $\map{T}_{\theta_0}^{(n)}$ and $\map{S}_{\theta_0}^{(n)}$, which depend on $\theta_0$ and $n$ (but not on the exact value of $\theta$). Using these two channels, $n$-copy qudit states and Gaussian states can be interconverted with an error vanishing in $n$ \cite{kahn-thesis}.  Explicitly, one has the following bounds
\begin{align}
\sup_{\theta\in\set{\Theta}_{n,x}(\theta_0,c)}\left\|\map{T}_{\theta_0}^{(n)}\left(\rho_{\theta}^{\otimes n}\right)-G_{n,\theta}\right\|_1=O\left(n^{-\kappa(x)}\right)\label{LAN-T}\\
\sup_{\theta\in\set{\Theta}_{n,x}(\theta_0,c)}\left\|\rho_{\theta}^{\otimes n}-\map{S}^{(n)}_{\theta_0}\left(G_{n,\theta}\right)\right\|_1=O\left(n^{-\kappa(x)}\right)\label{LAN-S},
\end{align}
where $\|\cdot\|_1$ denotes the trace norm and  $\kappa (x)$ is defined as
\begin{align}\label{kappa}
\kappa(x)=\min\left\{\frac{1-z-\eta}2,\frac{1-3x}{4}-y,\frac{2-9\eta}{24}\right\}
\end{align}
where $y,z,\eta$ can be freely chosen under the constraints $(1+x)/2<z<1$, $y>0$, $\eta>0$ and $\eta>x-y$.  With proper values for $y, z$, and $\eta$, when $x\in[0,2/9)$, the exponent $\kappa(x)$ is a non-increasing function of $x$ and falls within the interval $[0.027,0.084]$.

\item {\em Quantum state tomography.} State tomography is an important technique used to determine the density matrix of an unknown quantum state. In our protocol, the role of tomography  is to provide a rough estimate of $\theta_0$ so that we can apply Q-LAN. We adopt the  tomographic protocol proposed in  \cite{tomography}, which provides  an estimate $\rho_{\hat{\theta}}$ of a qudit state $\rho_\theta$.      When the protocol is carried out on $n$ copies of the state $\rho_\theta$, the estimate satisfies the bound
\begin{align}\label{prob-0}
\mathbf{Prob}\left[\frac12\|\rho_{\theta}-\rho_{\hat{\theta}}\|_1\le \varepsilon \right]\ge 1- (n+1)^{3d^2}e^{-n\varepsilon^2}
\end{align}
using $n$ copies of the state.

\end{itemize}

\begin{figure}  [b!]
\begin{center}
  \includegraphics[width=1\linewidth]{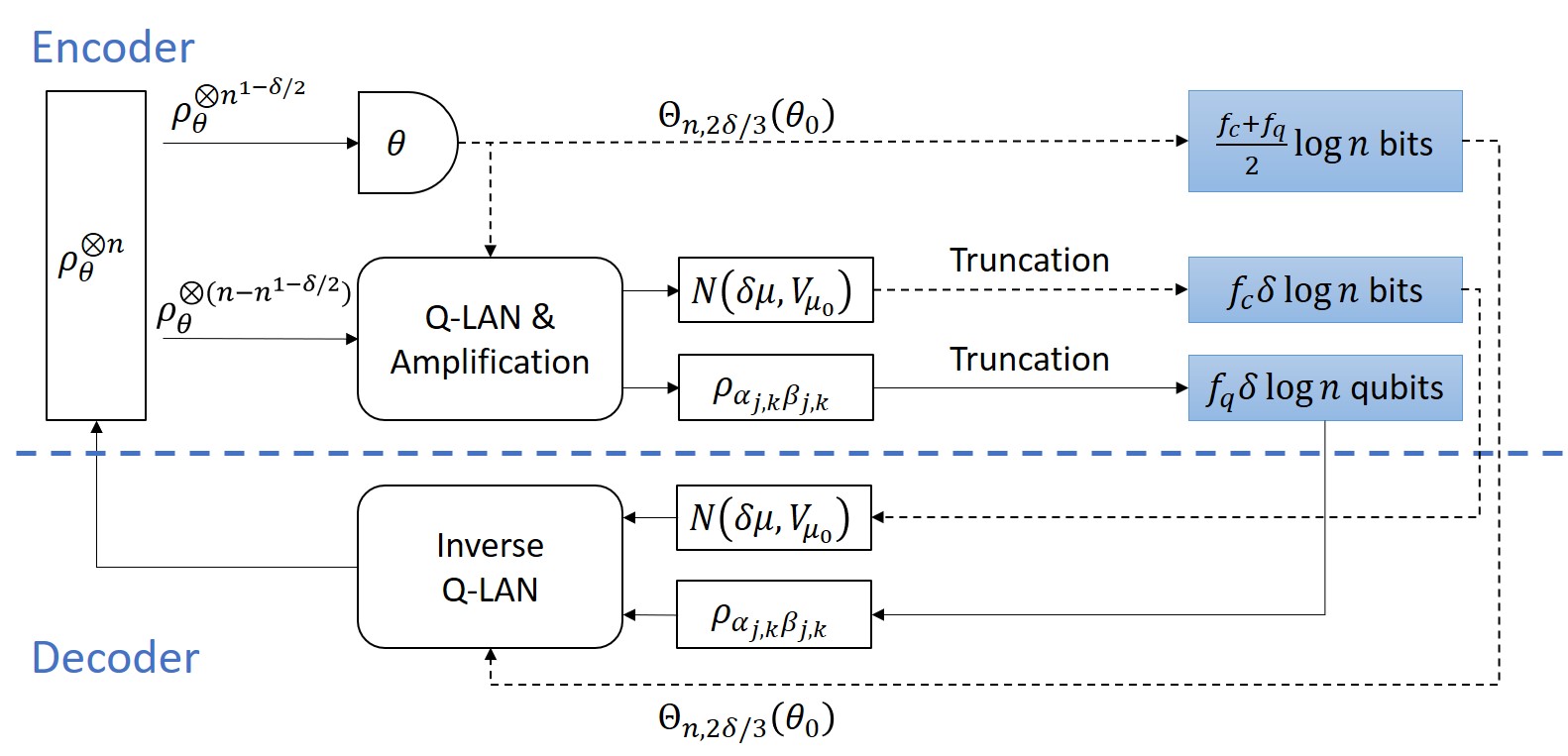}
  \end{center}
\caption{\label{fig-qudit}
  {\bf Compression protocol for qudit states.}  }
\end{figure}

Our compression protocol is illustrated in Fig. \ref{fig-qudit}.          For any $\delta\in(0,2/9)$, the protocol consists of the following steps:
\begin{itemize}
\item {\it Preprocessing.}
Divide  the parameter space $\set{\Theta}$ into a lattice $\set{L}:=\{\theta\in\set{\Theta}~|~\theta_i=z_i/(2\sqrt{n}), z_i\in\Z\ \forall\ i\}$. The lattice has approximately $  n^{(f_c+f_q)/2}$ points, which will be used to store the outcome of tomography.
\item {\it Encoder.} The encoder of $\rho_\theta^{\otimes n}$ consists of five steps:
\begin{enumerate}
\item {\em Tomography.} Use  $n^{1-\delta/2}$ copies of $\rho_\theta$  for quantum tomography, which yields an estimate $\hat{\theta}$ of $\theta$.
\item {\em Storage of the estimate.}    Encode the estimate $\hat{\theta}$ as a point in the lattice
\begin{align}\label{lattice}
\set{L}:=\{\theta\in\set{\Theta}~|~\theta_i=z_i/(2\sqrt{n}), z_i\in\Z\ \forall\ i\}.
\end{align}
Choose the lattice point $\theta_0$  that is closest to  $\hat{\theta}$, namely
\begin{align}\label{def-theta0}
\theta_0:=\argmin_{\theta'\in\set{L}}\|\hat{\theta}-\theta'\|_{\infty}.
\end{align}

\item {\em Q-LAN.} After the tomography step, we end up with $n-n^{1-\delta/2}$ copies of the state.  Define
\begin{align}\label{gamma-n1}
\gamma'_n=\frac{1}{1-n^{-\delta/2}}.
\end{align}
so that the number of remaining copies is $n/\gamma'_n$.
   The $n/\gamma'_n$ copies are sent through the channel $\map{T}_{\theta_0}^{(n/\gamma'_n)}$ (\ref{LAN-T}) which outputs the  Gaussian state $G_{(n/\gamma'_n),\theta}$ defined by Eq. (\ref{LAN-gaussian}).
   \item {\em Amplification.} To compensate the loss of copies in tomography, the state $\rho_{\alpha_{j,k},\beta_{j,k}}$ of each quantum mode is amplified by the amplifier defined in Eq. (\ref{amplifier}) with $\gamma=\gamma'_n$.
 The Gaussian distribution on the classical register is rescaled by a constant factor:
\begin{align}\label{amp-classical}
\map{A}_{\rm c}(\rho):=\sum_u \<u|\rho|u\>\,|\sqrt{\gamma'_n}u\>\<\sqrt{\gamma'_n}u|
\end{align}
 where $\{|u\>\<u|\}$ denotes the basis of the classical register with $u\in\R^{f_c}$. Notice that, in practice, $u$ can be approximated by a long sequence of classical bits with arbitrary high precision.

The whole amplification process is described by the channel $\map{A}^{(n/\gamma'_n)\to n}_{\theta_0}$ defined as the following:
\begin{align}\label{amp-total}
\map{A}^{(n/\gamma'_n)\to n}_{\theta_0}=\map{A}_{\rm c}\otimes \map{A}_{\rm q}\qquad \map{A}_{\rm q}=\bigotimes_{j<k}\map{A}^{\gamma'_n}.
\end{align}
where $\map{A}_{\rm c}$ is the classical amplifier just defined in Eq. (\ref{amp-classical}).

\item {\em Gaussian state compression.} Each quantum mode of the amplified Gaussian state is then truncated by $\map{P}_{n^\delta}$, defined by Eq. (\ref{trun-P-alpha}). The output state is then stored in a quantum memory of size $\delta \log n$ for each mode.
The classical mode is compressed by a map $\map{P}_{\rm c}$ that truncates the state into a $O(n^{\delta/2})$-hypercube centered around the mean of the Gaussian and rounds the continuous variable $u$ into a discrete lattice. Explicitly, we have
\begin{align}
\nonumber \map{P}_{\rm c}(\rho)&=\sum_{\|u\|_{\infty}\le n^{\delta/2}}\<u|\rho|u\>\,|r_{n,\delta}(u)\>\<r_{n,\delta}(u)|  \\
  &\quad +\left[1-\sum_{\|u'\|_{\infty}\le n^{\delta/2}}\<u'|\rho|u'\>\right]|0\>\<0|,
\end{align}
where $r_{n,\delta}(u)$ is the rounding function which maps $u\in \R^{f_c}$ to the closest point on the lattice $\left(\Z/n^{\delta/2}\right)^{f_c}$. The output of $\map{P}_c$ is stored in classical memory. The memory size is determined by the number of lattice points covered by the range of truncation. The separation between lattice points is $n^{-\delta/2}$, and the range of truncation is $[-n^{\delta/2}, n^{\delta/2}]^{f_c}$, so it covers $O(n^{f_c\delta})$ points on the lattice. We therefore need  $f_c \delta \log n$ bits.

The whole process for Gaussian state compression is described by the channel
\begin{align}
\map{P}^{(n)}_{\theta_0}&=\map{P}_{\rm c}\otimes \map{P}_{\rm q}\qquad \map{P}_{\rm q}=\bigotimes_{j<k}\map{P}_{n^\delta} \, .
\end{align}

\end{enumerate}
\item {\it Decoder.} Read $|t\>\<t|\in(\Z/n^{\delta/2})^{f_c}$ from the classical memory that stores the discretized Gaussian distribution and perform a uniform sampling in a shrinking hypercube containing $t$:
\begin{align}
\set{C}_{n,\delta}(t):=\left\{\hat{u}\in\R^{f_c}~|~\|\hat{u}-t\|_{\infty}\le (1/2)n^{-\delta/2}\right\}.
\end{align}
This step converts the discrete random variable $t$, whose distribution is a weighted sum of Dirac deltas, into a continuous random variable $\hat{u}$ with probability distribution that approximates $N\left(\delta\mu,V_{\mu_0}\right)$.
The output sample $\hat{u}$ together with the state of the quantum memory is sent through the channel $\map{S}_{\theta_0}^{(n)}$ (\ref{LAN-S}), which can be constructed from  the outcome of tomography.
\end{itemize}
\subsection{Error analysis.}
To bound the error of our protocol, we need to specify a small neighborhood for discussion, which should contain the true value $\theta$ with high probability. A proper choice is the neighborhood $\set{\Theta}_{n,2\delta/3}(\theta_0)$ (\ref{neighborhood}).
Using the triangle inequality of trace distance, we split the overall error into  four terms
\begin{align}
\epsilon\le & \epsilon_{\rm tomo}+\epsilon_{\rm amp}+\epsilon_{\rm G}+\epsilon_{\rm Q-LAN},
\end{align}
where
\begin{align}
\epsilon_{\rm tomo}&=\mathbf{Prob}\left[\theta\not\in\set{\Theta}_{n,2\delta/3}(\theta_0)\right]\\
\nonumber \epsilon_{\rm amp}&=\frac12\sup_{\theta_0}\sup_{\theta\in\set{\Theta}_{n,2\delta/3}(\theta_0)}\Big\|\map{A}^{(n/\gamma'_n)\to n}_{\theta_0}\left(G_{n/\gamma'_n,\theta}\right) \\ 
& \qquad  \qquad  \qquad \qquad \qquad-G_{n,\theta}\Big\|_1\label{amp-error}\\
\epsilon_{\rm G}&=\frac12\sup_{\theta_0}\sup_{\theta\in\set{\Theta}_{n,2\delta/3}(\theta_0)}\left\|\map{P}^{(n)}_{\theta_0}\left(G_{n,\theta}\right)-G_{n,\theta}\right\|_1\label{compress-error}\\
\nonumber \epsilon_{\rm Q-LAN}&=\frac12\sup_{\theta_0}\sup_{\theta\in\set{\Theta}_{n,2\delta/3}(\theta_0)}\Big\{\Big\|\map{T}^{(n/\gamma'_n)}_{\theta_0}\left(\rho_{\theta}^{\otimes (n/\gamma'_n)}\right)  \\
 \nonumber &\qquad  \qquad   \quad \qquad \qquad \qquad  \qquad-G_{n/\gamma'_n,\theta}\Big\|_1\\
  & \qquad   \qquad \qquad \qquad+\left\|\rho_{\theta}^{\otimes n}-\map{S}^{(n)}_{\theta_0}\left(G_{n,\theta}\right)\right\|_1\Big\}
\end{align}
are the error terms of tomography, amplification, truncation, and Q-LAN, respectively.  In the following, we will provide upper bounds for all  four terms.

Let us start from the tomography error.
By definition of the neighborhood $\set{\Theta}_{n,2\delta/3}(\theta_0)$ (\ref{neighborhood}), we have
\begin{align}
\epsilon_{\rm tomo}&=\mathbf{Prob}\left[\|\theta-\theta_0\|_{\infty}> n^{-1/2+\delta/3}\right]\nonumber\\
&\le\mathbf{Prob}\left[\|\theta-\hat{\theta}\|_{\infty}+\|\hat{\theta}-\theta_0\|_{\infty}>  n^{-1/2+\delta/3}\right]\nonumber\\
&\le\mathbf{Prob}\left[\|\theta-\hat{\theta}\|_{\infty}>n^{-1/2+\delta/3}(1-O(n^{-\delta/3}))\right].\label{prob-inter}
\end{align}
The first inequality comes from triangle inequality and the second inequality holds since $\|\hat{\theta}-\theta_0\|_{\infty}\le 1/(2\sqrt{n})$ which  is an immediate implication of Eq. (\ref{lattice}) and Eq. (\ref{def-theta0}).

To further bound the error, notice that the trace distance has a Euclidean expansion, of the form $\|\rho_\theta-\rho_{\theta'}\|_1 = C \|\theta-\theta'\|_{\infty}+O(\|\theta-\theta'\|_{\infty}^2)$, where $C>0$ is a suitable constant.  Then we have
\begin{align}
\epsilon_{\rm tomo}&\le\mathbf{Prob}\left[\frac12\|\rho_\theta-\rho_{\hat{\theta}}\|_1> (C/4) n^{-1/2+\delta/3}\right].
\end{align}
Substituting $\varepsilon$ with $(C/4) n^{-1/2+\delta/3}$ and $n$ with $n^{1-\delta/2}$ in Eq. (\ref{prob-0}) we have
\begin{align}
\epsilon_{\rm tomo}&=n^{-\Omega(n^{\delta/6})}\label{error-tomography}.
\end{align}

Next, we look at the error of amplification.
By Eqs. (\ref{LAN-gaussian}) and (\ref{amp-total}), the output of the amplifier can be expressed as
\begin{align}
\nonumber \map{A}^{(n/\gamma'_n)\to n}_{\theta_0}\left(G_{n/\gamma'_n,\theta}\right)= &\map{A}_{\rm c}\left( N\left(\delta\mu/\sqrt{\gamma'_n},V_{\mu_0}\right)\right)  \\
& \otimes\map{A}_{\rm q}\left(\Phi\left(\delta\xi/\sqrt{\gamma'_n},\mu_0\right)\right) \label{Gnovergamma}
\end{align}
where $(\delta\mu, \delta\xi) = \sqrt{n} (\theta-\theta_0)$.
We then split the amplification error (\ref{amp-error}) into two terms: the term of the classical mode and the term of quantum mode amplification.
\iffalse
We first analyze the classical mode, where the error comes from the mismatch of two Gaussian distributions $N\left(\delta\mu,V_{\mu_0}\right)$ and $N\left(\delta\mu/\sqrt{\gamma'_n},V_{\mu_0}\right)$.
This error can be bounded exploiting the Kullback-Leibler divergence between two Gaussian distributions (see, for instance, Eq. (6) in \cite{nielsen2009clustering}) and the Pinsker's inequality as:
\begin{align}\label{sb1}
\frac12\|N\left(\delta\mu,V_{\mu_0}\right)-N\left(\delta\mu/\sqrt{\gamma'_n},V_{\mu_0}\right)\|_1\le\frac{(\gamma'_n)^{-1/2}-1}2\sqrt{\delta\mu^T V_{\mu_0}^{-1} \delta\mu}.
\end{align}
From Eq. (\ref{neighborhood}) and Eq. (\ref{gamma-n1}) we get $\|\delta\mu\|_{\infty}=O(n^{\delta/3})$ and $(\gamma'_n)^{-1/2}-1=O(n^{-\delta/2})$, respectively. Substituting into Eq. (\ref{sb1}), we get
\begin{align}\label{sb2}
\frac12\|N\left(\delta\mu,V_{\mu_0}\right)-N\left(\delta\mu/\sqrt{\gamma'_n},V_{\mu_0}\right)\|_1=O\left(n^{-\delta/ 6}\right).
\end{align}
\fi
We first analyze the classical mode, where the amplifier rescales the classical Gaussian distribution.
The amplifier $\map{A}_{\rm c}$ (\ref{amp-classical}) amplifies the distribution by a factor of $\sqrt{\gamma'_n}$, shifts the center of the Gaussain distribution from $\delta\mu/\sqrt{\gamma'_n}$ to $\delta\mu$, and rescales the covariance matrix by a factor of $\gamma'_n$, from $V_{\mu_0}$ to $\gamma'_n V_{\mu_0}$:
\begin{align}
\map{A}_{\rm c}\left(N\left(\delta\mu/\sqrt{\gamma'_n},V_{\mu_0}\right)\right) = N\left(\delta\mu,\gamma'_nV_{\mu_0}\right).
\end{align}
%{\color{red} The previous sentence is very hard to parse. I am completely lost here. }As a result, we consider the difference of the following distributions: $N_{\delta\mu,I_{\mu_0}}$ and $N_{\delta\mu,\gamma'_nI_{\mu_0}}$. As they have the same center, we may translate them both to the origin. The error  {\color{red}  What error is this?  } is:

The amplification error for the classical mode is:
\begin{align}
\nonumber \epsilon_{\rm classical}
& =  \|N\left(\delta\mu,\gamma'_nV_{\mu_0}\right) - N\left(\delta\mu,V_{\mu_0}\right)\|_1 \\
& =  \|N\left(0,\gamma'_nV_{\mu_0}\right) - N\left(0,V_{\mu_0}\right)\|_1.
\end{align}
Writing explicitly the probability density functions of the Gaussian distributions, we have
\begin{align}
\nonumber \epsilon_{\rm classical}& =  \int_{\mathbb{R}^d} \Big|
\frac{1}{\sqrt{2\pi |V_{\mu_0}|}}\exp\left(-\frac{\mathbf{x}^T V_{\mu_0}^{-1} \mathbf{x}}2 \right) \\
\nonumber &  \quad  -
\frac{1}{\sqrt{2\pi \gamma'_n |V_{\mu_0}|}}\exp\left(-\frac{\mathbf{x}^T V_{\mu_0}^{-1} \mathbf{x}}{2\gamma'_n} \right)
\Big| d\mathbf{x} \\
\nonumber &\le \frac{1}{\sqrt{2\pi |V_{\mu_0}|}} \int_{\mathbb{R}^d} \left|\exp\left(-\frac{\mathbf{x}^T V_{\mu_0}^{-1} \mathbf{x}}2 \right)\right. \\
\nonumber &  \quad   \qquad \qquad\qquad -\left. \exp\left(-\frac{\mathbf{x}^T V_{\mu_0}^{-1} \mathbf{x}}{2\gamma'_n} \right) \right| d\mathbf{x} \\
\nonumber &+ \frac{\left(1-(\gamma'_n)^{-1/2}\right)}{\sqrt{2\pi|V_{\mu_0}|}} \int_{\mathbb{R}^d} \exp\left(-\frac{\mathbf{x}^T V_{\mu_0}^{-1} \mathbf{x}}{2\gamma'_n}\right)d\mathbf{x}\\
\nonumber &\le \frac{(1-(\gamma'_n)^{-1})}{2\sqrt{2\pi |V_{\mu_0}|}} \int_{\mathbb{R}^d}   \mathbf{x}^T V_{\mu_0}^{-1}\mathbf{x}\exp\left(-\frac12\mathbf{x}^T V_{\mu_0}^{-1}\mathbf{x}\right) d\mathbf{x}  \\
\nonumber &  \quad  + O\left(n^{-\delta/2}\right)\\
%\leq & O(n^{-\delta/2}) \int_{\mathbb{R}^d} \mathbf{x}^T V_{\mu_0}^{-1}\mathbf{x} e^{-\frac12\mathbf{x}^T V_{\mu_0}^{-1}\mathbf{x}} d\mathbf{x}\\
&=  O\left(n^{-\delta/2}\right).
\end{align}
Note that $(1-(\gamma'_n)^{-1})$ and $(1-(\gamma'_n)^{-1/2})$ both have order $O(n^{-\delta/2})$.

Now we check the quantum term. On the quantum register, the amplifier acts independently on each mode as the displaced thermal state amplifier defined by Eq. (\ref{amplifier}).
From a similar calculation as Eq. (\ref{amp-output}), we obtain the inequality
\begin{align}
\epsilon_{\rm quantum}\le\frac12\sum_{j<k}\left\|\map{A}^{\gamma'_n}\left(\rho_{\alpha_{j,k},\beta_{j,k}}\right)-\rho_{\alpha_{j,k},\beta_{j,k}}\right\|_1,
\end{align}
where the error of each quantum amplifier is [cf. Eqs. (\ref{mancante}) and (\ref{gamma-n1})]
\begin{align}
\frac12\left\|\map{A}^{\gamma'_n}\left(\rho_{\alpha_{j,k},\beta_{j,k}}\right)-\rho_{\alpha_{j,k},\beta_{j,k}}\right\|_1=O\left(n^{-\delta/2}\right).\label{amplification-error}
\end{align}
Therefore, we conclude that the amplification error (\ref{amp-error}) scales at most as
\begin{align}\label{error-amplification}
\epsilon_{\rm amp}&\le \epsilon_{\rm classical}+\epsilon_{\rm quantum} = O\left(n^{-\delta/2}\right).
\end{align}

Let us now consider the error (\ref{compress-error}) of the Gaussian state compression, which can be upper bounded as
\begin{align}
\nonumber \epsilon_{\rm G}\le &\frac12\left\|\map{P}_{\rm c}\left[N(\delta\mu,V_{\mu_0})\right]-N\left(\delta\mu,V_{\mu_0}\right)\right\|_1\\
\label{error-truncation0}
 &  +\frac12\sum_{j<k}\left\|\map{P}_{n^\delta}\left(\rho_{\alpha_{j,k},\beta_{j,k}}\right)-\rho_{\alpha_{j,k},\beta_{j,k}}\right\|_1
\end{align}
For the classical part, there are two sources of error: the error of rounding  and the error of truncation. The former is simply $O(n^{-\delta/2})$, equal to the resolution of the rounding. The latter can be bounded by
noticing that $\|\delta\mu\|_{\infty}\le n^{\delta/3}$, from which we have
\begin{align}
\nonumber &\left\|\map{P}_{\rm c}\left[N(\delta\mu,V_{\mu_0})\right]-N\left(\delta\mu,V_{\mu_0}\right)\right\|_1\\
&  \le\int_{\|u\|_{\infty}> n^{\delta/2}}N\left(\delta\mu,V_{\mu_0}\right)(\d u)\nonumber\\
&\le\int_{\|u-\delta\mu\|_{\infty}> n^{\delta/2}-n^{\delta/3}}N\left(\delta\mu,V_{\mu_0}\right)(\d u)\nonumber\\
&=e^{-\Omega(n^{\delta})}\label{error-truncation-classical}
\end{align}
where $N\left(\delta\mu,V_{\mu_0}\right)(\d u)$ denotes the probability density function. For each of the quantum modes,
employing Lemma \ref{lemma-truncation}, with $K$ substituted by $n^\delta$ and $|\alpha_{j,k}|=O(n^{\delta/3})$, we have
\begin{align}\label{error-truncation-quantum}
\frac12\left\|\map{P}_{j,k}\left(\rho_{\alpha_{j,k},\beta_{j,k}}\right)-\rho_{\alpha_{j,k},\beta_{j,k}}\right\|_1=\beta_{j,k}^{\Omega(n^{\delta/8})}+e^{-\Omega(n^{\delta/4})}.
\end{align}

Substituting Eqs. (\ref{error-truncation-classical}) and (\ref{error-truncation-quantum}) into Eq. (\ref{error-truncation0}), we have
\begin{align}\label{error-truncation}
\epsilon_{\rm G}= \left(\max_{j<k}\beta_{j,k}\right)^{\Omega(n^{\delta/8})}+e^{-\Omega(n^{\delta/4})}.
\end{align}

Finally, we note that the error of the Q-LAN approximation, corresponding to the errors generated by the transformations between the input state and its Gaussian state approximation, is given by Eqs. (\ref{LAN-T}) and (\ref{LAN-S}) as $\epsilon_{\rm Q-LAN}=O\left(n^{-\kappa(2\delta/3)}\right)$. Since $\kappa$ is non-increasing, the bound can be relaxed to
\begin{align}
\epsilon_{\rm Q-LAN}&=O\left(n^{-\kappa(\delta)}\right),\label{error-lan}
\end{align}
 Summarizing the above bounds (\ref{error-tomography}), (\ref{error-amplification}), (\ref{error-truncation}), (\ref{error-lan}) on each of the error terms, we conclude that the protocol generates an error which scales at most
\begin{align}
\epsilon=O\left(n^{-\kappa(\delta)}\right)+O\left( n^{-\delta/2}\right).
\end{align}

\subsection{Total memory cost.}
The total memory cost consists of three  contributions: a classical memory of $[(f_c+f_q)/2]\log n$ bits for the tomography outcome, a classical memory of $f_c\delta\log n$ bits for the classical part of the Gaussian state and a quantum memory of $f_q\delta\log n $ qubits for the quantum part of the Gaussian state. In short, it takes $(1/2+\delta)\log n$ bits to encode a classical independent parameter and $(1/2)\log n$ bits plus $\delta\log n$ qubits to encode a quantum independent parameter.

From the above discussion we can see that the ratio between the quantum memory cost and the classical memory cost is
\begin{align}\label{ratio}
R_{q/c}=\frac{\delta f_q}{(1/2+\delta)f_c+(1/2)f_q},
\end{align}
which can be made close to zero by choosing $\delta$  close to zero.   In conclusion, the size of the quantum memory can be made arbitrarily small compared to the classical memory.

\section{Necessity of a quantum memory.}\label{sec:classical}
In the previous section, we showed that the ratio between the quantum and the classical memory cost can be made arbitrarily close to zero [see Eq. (\ref{ratio})]. It is then natural to ask whether the ratio can   be  exactly zero.
% i.e. compressing faithfully using a fully classical memory.
The answer turns out to be  negative. In fact, we prove an even stronger result:  if a state family has  at least one independent quantum parameter, then no protocol using a purely classical memory can be faithful, even if the amount of classical memory is arbitrarily large.
\begin{mythm}\label{theo:notonlybits}
Let  $\{\rho_{\theta}^{\otimes n}\}$ be a qudit state family with at least one independent quantum parameter, and let  $(\map{E}_n,\map{D}_n)$ be generic compression protocol for $\{\rho_{\theta}^{\otimes n}\}$.    If the protocol uses solely a classical memory, then the compression error will not vanish in the large $n$ limit, no matter how large the memory is.
\end{mythm}

The proof of Theorem  \ref{theo:notonlybits}  is based on the properties of  two distance measures, known as the quantum Hellinger distance \cite{hellinger} (see also \cite{luo})  and the Bures distance  \cite{nielsen-chuang}, and defined as
\begin{align}
d_{\rm H}(\rho_1,\rho_2)&:= \sqrt{2-2\Tr \left(\rho_1^{1/2}\rho_2^{1/2}\right)}\\
d_{\rm B}(\rho_1,\rho_2)&:=\sqrt{2- 2\Tr \left|\rho_1^{1/2}\rho_2^{1/2}\right|} \, ,
\end{align}
respectively.
% $d_{\rm H}$ is known as the quantum Hellinger distance and $d_{\rm B}$ is known as the Bures distance.
%The proof that $d_{\rm H}$ and $d_{\rm B}$ are non-negative and satisfy the triangular inequality is provided in Refs.  {\color{red}  add references.}

The first property used in the proof of Theorem \ref{theo:notonlybits}  is
\begin{Lemma}\label{lemma-commute}
For every pair of density matrices $\rho_1$ and $\rho_2$, one has the inequality
$d_{\rm H}(\rho_1,\rho_2)\ge d_{\rm B}(\rho_1,\rho_2)$.  The equality holds if and only if $[\rho_1,\rho_2]=0$.
\end{Lemma}

\begin{proof}
By definition, the condition $d_{\rm H}(\rho_1,\rho_2)\ge d_{\rm B}(\rho_1,\rho_2)$ is equivalent to the condition  $$\Tr \left|\rho_1^{1/2}\rho_2^{1/2}\right|\ge\Tr \left(\rho_1^{1/2}\rho_2^{1/2}\right) \, .$$
 The validity of this condition is immediate:  for every square matrix $A$, one has $\Tr |A|  \ge \Tr A$.   The equality holds if and only if $A$ is equal to $|A|$, meaning that $A$ is positive semidefinite.     For $A   =   \rho_1^{1/2}\rho_2^{1/2}$, the  Hermiticity requirement  $A  =  A^\dag$ reads
  $$\left(\rho_1^{1/2}\rho_2^{1/2}\right)     =   \left(\rho_1^{1/2}\rho_2^{1/2}\right)^\dag=\left(\rho_2^{1/2}\rho_1^{1/2}\right)\, ,$$   which in turn is equivalent to the commutation relation  $[\rho_1,\rho_2]=0$.
\end{proof}

The second property used in the proof of Theorem  \ref{theo:notonlybits}  is
\begin{Lemma}\label{lemma-continuity}
 Let $\map{E}$ be a quantum channel sending states on $\spc H$ to states on  $\spc K$ and let $\map D$ be a quantum channel sending states on $\spc K$ to states on $\spc H$.     Let  $\rho_1$ and $\rho_2$ be two states on  $\spc{H}$, satisfying the conditions
 \begin{align}\label{encodcommut}  [  \map{E}(\rho_1),\map{E}(\rho_2)]=0  \, ,
 \end{align}
 and $\frac12\|\map{D}\circ\map{E}(\rho_i)-\rho_i\|_1\le \epsilon$  for $i  \in  \{1,2\}$.
Then, the following inequality holds:
\begin{align}\label{tobeproven}
\left|d_{\rm H}(\rho_1,\rho_2) - d_{\rm B}(\rho_1,\rho_2)\right|\le2\sqrt{2\epsilon}.
\end{align}
\end{Lemma}

\begin{proof}
Using the  the triangle inequality for the quantum Hellinger distance, we obtain the upper bound
\begin{align}  \nonumber 
 d_{\rm H}(\rho_1,\rho_2)&\le   d_{\rm H}\left(\rho_1,\map{D}\circ\map{E}(\rho_1)\right)   +  d_{\rm H}\left(\map{D}\circ\map{E}(\rho_1),\map{D}\circ\map{E}(\rho_2)\right) \\
\label{aaa} &  \quad  +d_{\rm H}\left(\map{D}\circ\map{E}(\rho_2),\rho_2\right)  \,.
\end{align}
Now, for every pair of states $\rho$ and $\sigma$, the quantum Hellinger distance and the trace distance are related by inequality   $d_{\rm H}(\rho,\sigma)\le \sqrt{\|\rho-\sigma\|_1}$ \cite{hellinger}.   Using this  fact, the upper bound   (\ref{aaa}) becomes
\begin{align}\label{bbb} d_{\rm H}(\rho_1,\rho_2)\le    d_{\rm H}\left(\map{D}\circ\map{E}(\rho_1),\map{D}\circ\map{E}(\rho_2)\right)    +   2  \sqrt{2\epsilon} \,.
\end{align}

At this point, we use  the fact that  the quantum Hellinger (respectively, Bures) distance is  non-increasing under the action of quantum channels   \cite{luo}  (respectively, \cite{nielsen-chuang}).
In this way, we obtain  the inequality
\begin{align}
\nonumber d_{\rm H}\left(\map{D}\circ\map{E}(\rho_1),\map{D}\circ\map{E}(\rho_2)\right)  &\le       d_{\rm H}\left(\map{E}(\rho_1),\map{E}(\rho_2)\right) \\
 \nonumber  & =d_{\rm B}\left(\map{E}(\rho_1),\map{E}(\rho_2)\right)   \\
   &\le  d_{\rm B}\left(\rho_1,\rho_2\right) \, , \label{ccc}
\end{align}
in which we used the relation
\begin{align}
d_{\rm H}\left(\map{E}(\rho_1),\map{E}(\rho_2)\right)=d_{\rm B}\left(\map{E}(\rho_1),\map{E}(\rho_2)\right) \, ,
\end{align}
following  from Eq. (\ref{encodcommut})  and  Lemma \ref{lemma-commute}.   Combining Eqs. (\ref{bbb}) and (\ref{ccc}), we finally obtain the bound
\begin{align}
d_{\rm H}(\rho_1,\rho_2)\le d_{\rm B}\left(\rho_1,\rho_2\right)+2\sqrt{2\epsilon} \, .
\end{align}
Since the difference between the quantum Hellinger distance and the Bures distance is non-negative, the above inequality is exactly  Eq. (\ref{tobeproven}).
\end{proof}

Now we give the proof of Theorem \ref{theo:notonlybits}.
\begin{proof}
Let  $\{\rho_\theta^{\otimes n}\}$ be a qudit state family with at least one independent quantum parameter.

Pick two states $\rho_{\theta_0}^{\otimes n}$ and $\rho_{\theta}^{\otimes n}$, with $\theta$ of the form $\theta=\theta_0+s/\sqrt{n}$  where all entries of the vector $s$ are zero except for the independent quantum parameter. Applying Q-LAN to the neighborhood of $\theta_0$, the two states $\rho_{\theta_0}^{\otimes n}$ and $\rho_{\theta}^{\otimes n}$ can be converted into two multi-mode Gaussian states $G_{n,\theta_0}$ and $G_{n,\theta}$ that differ from each other only in one mode. Explicitly, the two Gaussian states can be written as
\begin{align}\label{ddd}
G_{n,\theta_0}&=G^{(-)}_{n,\theta_0}\otimes\rho_{\beta}^{{\rm thm}} \\
\label{eee} G_{n,\theta}&=G^{(-)}_{n,\theta_0}\otimes\rho_{\alpha(s),\beta} \, .
\end{align}
where the thermal parameter  $\beta$ and the displacement $\alpha(s)$ are non-zero quantities depending only on $s$ and $\theta_0$ via Eq. (\ref{LAN-gaussian}), while  $G^{(-)}_{n,\theta_0}$ is the state of the remaining modes.

   Now,  let  $(\map{E}_{n},\map{D}_{n})$ be a compression protocol that uses  a purely classical memory to compress the state family $\{\rho_\theta^{\otimes n}\}$.    By Q-LAN, there exists a compression protocol $(\map{E}'_{n},\map{D}'_{n})$ that uses a purely classical memory to compress the states  $\{G_{n,\theta_0},G_{n,\theta}\}$.
 Explicitly, the encoder and the decoder are described by the channels   % using also a fully classical memory, where
\begin{align}
\map{E}'_{n}:=\map{E}_{n}\circ\map{S}_{\theta_0}^{(n)}\qquad \map{D}'_{n}:=\map{T}_{\theta_0}^{(n)}\circ\map{D}_{n} \, ,
\end{align}
where    $\map{T}_{\theta_0}^{(n)}$ and  $\map{S}_{\theta_0}^{(n)}$ are the channels  used for Q-LAN.

Since the protocol  $(\map{E}'_{n},\map{D}'_{n})$  uses a purely classical memory,  Lemma \ref{lemma-continuity} implies the bound
\begin{align}
\left|d_{\rm H}\left(G_{n,\theta_0},G_{n,\theta}\right) - d_{\rm B}\left(G_{n,\theta_0},G_{n,\theta}\right)\right|\le2\sqrt{2\epsilon'_n}\label{commutators},
\end{align}
where $\epsilon'_n$ is the compression error for  the states  $\{G_{n,\theta_0},G_{n,\theta}\}$.

On the other hand, the errors from Q-LAN vanish as  $O\left(n^{-\kappa(x)}\right)$.  Hence, we have the bound
\begin{align}\label{fff} \epsilon'_n\le\epsilon_n+O\left(n^{-\kappa(x)}\right)
\end{align}

Substituting Eqs.   (\ref{ddd}) and  (\ref{eee}) into Eq. (\ref{commutators}),  we obtain the expression
\begin{align}
\epsilon_n&\ge\epsilon'_n-O\left(n^{-\kappa(x)}\right)\nonumber\\
&\ge \frac18\left|d_{\rm H}\left(G_{n,\theta_0},G_{n,\theta}\right) - d_{\rm B}\left(G_{n,\theta_0},G_{n,\theta}\right)\right|^2-O\left(n^{-\kappa(x)}\right)\nonumber\\
&=\frac18\left|d_{\rm H}\left(\rho_{\beta}^{{\rm thm}},\rho_{\alpha(s),\beta}\right) - d_{\rm B}\left(\rho_{\beta}^{{\rm thm}},\rho_{\alpha(s),\beta}\right)\right|^2-O\left(n^{-\kappa(x)}\right)\label{commutators2}.
\end{align}
The quantum Hellinger and Bures distances between displaced thermal states can be computed using previous results. Using Eq. (16) of \cite{caonima1}, we have
\begin{align}
d_{\rm B}\left(\rho_{\beta}^{{\rm thm}},\rho_{\alpha,\beta}\right)=\sqrt{2-2e^{-\frac{|\alpha|^2}{4\gamma_{\rm B}}}}\qquad \gamma_{\rm B}=\frac{1+\beta}{1-\beta}.
\end{align}
Using Eq. (3.18) of \cite{caonima2}, we have
\begin{align}
d_{\rm H}\left(\rho_{\beta}^{{\rm thm}},\rho_{\alpha,\beta}\right)=\sqrt{2-2e^{-\frac{|\alpha|^2}{4\gamma_{\rm H}}}}\qquad \gamma_{\rm H}=\frac{(\sqrt{\beta}+1)^2}{2(1-\beta)}.
\end{align}
The difference between the two terms is strictly positive, except in the case when $\alpha=0$. Therefore, the right hand side of Eq. (\ref{commutators2}) is strictly positive in the limit of $n\to\infty$, and thus $\lim_{n\to\infty}\epsilon_n>0$. This concludes the proof.
\end{proof}

\section{Optimality of the compression}\label{sec-optimality}
Here we prove that our compression protocol is asymptotically optimal in terms of total memory cost.   Specifically, we show  that every   compression protocol    with vanishing error must use an overall memory size of  at least $(f/2)\log n$, where $f$ is the number of independent parameters describing the input states.

The proof idea is to construct a communication protocol that transmits approximately $(f/2)\log n$ bits, using the compression protocol.
    Once this is done,  the Holevo bound  \cite{holevo-1973}  implies that the overall amount of memory must be of at least $(f/2)\log n$ qubits/bits.

To construct the communication protocol, we define a mesh $\set{M}_n$ on the parameter space $\set{\Theta}$, by choosing a set of equally spaced points starting from a fixed point $\theta_0  \in \Theta$.  Specifically, we define the mesh as
\begin{align}
\set{M}_n=\left\{\theta\in\set{\Theta}~|~|(\theta-\theta_0)_i|=z_i\cdot\log n/\sqrt{n}, z_i\in\mathbb{Z} \ \forall\ i\right\} \, .
\end{align}
 The number of points  in the mesh  $\set M_n$ satisfies the bound
\begin{align}\label{points}
|\set{M_n}|\ge T_{\set{\Theta}}\left[\frac{\sqrt{n}}{\log n}\right]^{f},
\end{align}
where $T_{\set{\Theta}}>0$ is a constant independent of $n$.

The next step is to define a finite set of states
\begin{align}
\set S_n  =  \left\{\rho_\theta^{\otimes n}  ~|~  \theta\in \set M_n \right\} \, ,
\end{align}
and to observe that they are almost perfectly distinguishable in the large $n$ limit.    One way to distinguish between the states  in  $\set S_n$ is to use quantum tomography.  Intuitively, since tomography provides an estimate of the state with error of order $1/\sqrt n$,    the distance between two states in the set  is large enough to make the states almost perfectly distinguishable.
   To make this argument rigorous, we  describe the tomographic protocol using a POVM $P_n  (\d \hat \rho)$, where $\hat \rho $ is the estimate of the state.    In particular, we use the POVM defined in Eq. (\ref{prob-0}), which has the property  \cite{tomography}
\begin{align}
\int_{\|\hat \rho-\rho\|_1  >\varepsilon}\,\Tr[P_n (\d \hat \rho)  \,\rho^{\otimes n}]&\le (n+1)^{3d^2}e^{-n\varepsilon^2}.
\end{align}

The continuous POVM    $P_n  (\d \hat \rho )$ can be used to distinguish between the states in the set $\set S_n$.    To this purpose, we  construct the discrete POVM  with operators   $\{  Q_n (\theta) \}_{\theta  \in  \set M_n}   \cup  \{ Q_n ({\rm rest})\}$, where the operator  $Q_n ({\rm rest})$ is  defined as
\begin{align}
Q_n(  {\rm rest})  =    I^{\otimes n}   -   \sum_{\theta \in \set M_n}  \,    Q_n ( \theta) \,
\end{align}
while the operator $Q_n (\theta)$ is defined as
\begin{align}
Q_n (\theta)&=\int_{ \|\hat \rho-\rho_{{\theta}}\|_1\le  \frac {\varepsilon_{\min}}2 }  \,  P_n (\d \hat \rho)  \,  ,  \qquad      \qquad    \theta  \in \set M_n  \, .
\end{align}
Here $\varepsilon_{\min}$ is the minimum distance between two distinct states in $\set S_n$, which can be quantified as
\begin{align}
\varepsilon_{\min}=\min_{\theta, \theta'\in\set{M}_n \, ,   \theta \not = \theta'}\frac12\|\rho_\theta-\rho_{\theta'}\|_1=\frac{C\log n}{2\sqrt{n}}+O\left(\frac{\log^2 n}{n}\right) \, ,
\end{align}
having used the   Euclidean expansion of trace distance, given by
\begin{align}\|\rho_\theta-\rho_{\theta'}\|_1 = C \|\theta-\theta'\|_\infty+O(\|\theta-\theta'\|_\infty^2) \, ,
\end{align}
where  $C>0$ is a suitable constant.  Hence, for any $\theta \in  \set M_n $ the probability of error for the state $\rho_{\theta}^{\otimes n}$ can be bounded as
\begin{align}
\nonumber P_{{\rm err},  n}(\theta)&\le  \int_{\|\hat \rho  - \rho_{\theta}\|_1>\frac{\varepsilon_{\min}}2} \, \Tr\left[  P_n  (   \d \hat \rho ) \rho_{\theta}^{\otimes n}  \right]\\
\nonumber &\le (n+1)^{3d^2}e^{-\frac{C\log^2 n}{16}}\\
\nonumber &= (n+1)^{3d^2}n^{-\frac{C\log n}{16\ln 2}}\\
\label{perr}  &\le n^{-\frac{C\log n}{16}} ,
\end{align}
where the last inequality holds for large enough $n$.

Using the results above, we can construct a communication protocol that communicates $(f/2)  \log n$ bits given any compression protocol $(\map{E}_{n},\map{D}_{n})$.  The protocol is defined as follows:
\begin{enumerate}
\item Both parties agree   on a code that associates  messages with points in the mesh $\set{M}_n$.
\item To communicate a certain message, the sender picks the corresponding point $\theta\in\set M_n$   and prepares the state $\rho_\theta^{\otimes n}$.
\item The sender applies the encoder $\map{E}_{n}$ and transmits $\map{E}_n(\rho_\theta^{\otimes n})$ to the receiver.
\item The receiver applies the  decoder $\map{D}_{n}$.
\item The receiver   measures the output state with the POVM $\{Q_n (\theta) \}_{\theta \in\set M_n}  \cup  \{  Q_n ({\rm rest})\}$.
\end{enumerate}

The protocol is illustrated in Fig. \ref{fig-opt}. A protocol can be constructed for displaced thermal states, following the steps 1) - 4) and replacing the POVM in step 5) of the above protocol by the heterodyne measurement of $\alpha$ and maximum likelihood estimation of $\beta$ \cite{BS}. In this way, the proof here can be converted to a proof of optimality for displaced thermal states, which we omit for simplicity.
\begin{figure}  [t!]
\begin{center}
  \includegraphics[width=1\linewidth]{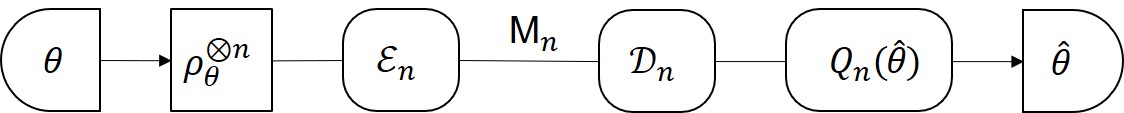}
  \end{center}
\caption{\label{fig-opt}
  {\bf A protocol to communicate $\log|\set{M}_n|$ bits of information.} Here $\map{E}_{n}$ and $\map{D}_{n}$ are the encoder and the decoder, $Q_n  (\hat \theta)$ is the POVM to recover the message, and $\set{M}_n$ denotes the memory. }
\end{figure}

It is not hard to see that the protocol can communicate no less than  $(f/2)  \log n$ bits, with an error probability
\begin{align}\label{P-ast}
P^\ast_{{\rm err}, n}\le P_{{\rm err}, n}+\epsilon_n
\end{align}
where $P_{{\rm err}, n} :  =  \max_{\theta \in \set M_n}      P_{{\rm err}, n}  (\theta)$ and $\epsilon_n$ is the error of the compression protocol $(\map{E}_{n},\map{D}_{n})$.
 Consider the case when the messages are uniformly distributed.  In this case, the number of transmitted bits can be bounded through Fano's inequality, which yields the bound
\begin{align}
I(\Theta:\hat{\Theta}) &\ge (1-P^\ast_{{\rm err}, n})\log|\set{M}_n|-h(P^\ast_{{\rm err},  n}) \,,
\end{align}
with $h(x)=-x\log x$ and $h(0):=0$.
When the compression protocol has vanishing error, i.e. $\lim_{n\to\infty}\epsilon_n=0$,
using Eqs. (\ref{points}), (\ref{perr}), and (\ref{P-ast}) we obtain the lower bound
\begin{align}
I(\Theta:\hat{\Theta}) &\ge \frac{f}2\log n-f\log\log n+o(1)\label{opt-inter2} \, .
\end{align}
Using the monotonicity of mutual information and the upper bound of entropy, the total number $n_{\rm enc}$ of memory bits/qubits is lower bounded as
\begin{align}
n_{\rm enc}&\ge H\left[\map{E}_{n}\left(|\set{M}_n|^{-1}\sum_{\theta\in\set{M}_n}\rho_\theta^{\otimes n}\right)\right]\nonumber \\
&\ge I(\Theta: \set{M}_n)\nonumber\\
& \ge I(\Theta:\hat{\Theta})\label{opt-inter3}.
\end{align}
Combining Eq. (\ref{opt-inter3}) with Eq. (\ref{opt-inter2}), we obtain that $n_{\rm enc}$ must be at least
\begin{align}
n_{\rm enc}&\ge  \frac{f}2\log n-f\log\log n+o(1) \, .
\end{align}
This proves that $f/2  \log n$  bits/qubits are necessary to achieve compression with vanishing error.

%\section{Necessity to use quantum memory}\label{sec-fisher}

\section{Conclusion}\label{sec-discussion}

In this work we addressed the problem of compressing identically prepared states of finite-dimensional quantum systems and identically prepared displaced thermal states. We showed that
%the total size of the required memory is proportional to the number of independent parameters of the state.
the total size of the required memory is approximately $(f/2) \log n$, where $f$ is the number of independent parameters of the state and $n$ is the number of input copies. Moreover, we  observed that the asymptotic ratio between the amount of quantum bits and the amount of classical  bits can be set to an arbitrarily small constant.  Still,  a  fully classical memory cannot faithfully encode genuine quantum states: only states that are jointly diagonal in fixed basis can be compressed into a purely classical memory.

A natural development of our work is the study of compression protocols for  quantum  population coding beyond the case of identically prepared states.  Motivated by the existing literature on classical population coding \cite{WAN}, the idea is to consider families of states representing a  population of  quantum particles.  At the most fundamental level, the indistinguishability of quantum particles leads to the  Bose-Einstein \cite{bose1924plancks,einstein1924quantentheorie} and Fermi-Dirac statistics \cite{fermi1928statistische,dirac1928quantum}, as well as to other intermediate statistics  \cite{leinaas1977theory,wilczek1982quantum}.   Since the Hilbert space describing identical quantum particles is not the
  tensor product of single-particle Hilbert spaces, the compression of quantum populations of indistinguishable particles  requires a non-trivial extension of our results.   The optimal compression protocols  are likely to shed light on how information is encoded into a broad range of real physical systems. In addition, the compression protocols will offer a tool to  simulate  large numbers of particles using quantum computers of relatively smaller size.

From the point of view of quantum simulations,
it is also meaningful to consider the compression of  tensor network states \cite{tensor1,tensor2} %\cite{fannes1992finitely,verstraete2004renormalization,vidal2007entanglement}
 which  provide a variational ansatz for a large number of quantum manybody systems.
The extension of quantum compression to this scenario is appealing as a technique to reduce the number of qubits needed to simulate systems of distinguishable quantum particles, in the same spirit of the compressed simulations introduced by Kraus \cite{kraus2011compressed} and coauthors  \cite{boyajian2013compressed,boyajian2015compressed}.
In the long term, the information-theoretic study of manybody quantum systems may provide a new approach to the simulation of complex systems  that are not efficiently simulatable on classical computers.

%Also, we did not discuss the case with fully classical memory.
%It is another interesting open problem to clarify the difference between the performances of quantum and classical memories.
%We further showed that fully classical memory cannot faithfully encode genuine quantum states, which distinguished the work from a simple extension of the classical version of this problem.

\subsection*{Acknowledgments}

This work is supported by the Canadian Institute for Advanced Research
(CIFAR), by the Hong Kong Research Grant Council
through Grants No. 17326616 and 17300317, by National Science Foundation of China through Grant No. 11675136, and
by the HKU Seed Funding for Basic Research, and by the Foundational
Questions Institute through grant FQXi-RFP3-1325.   Y. Y. is
supported by a Hong Kong and China Gas Scholarship.
MH was supported in part by a MEXT Grant-in-Aid for Scientific Research (B) No. 16KT0017, the Okawa Research Grant and Kayamori
Foundation of Informational Science Advancement.

\bibliography{ref}
\bibliographystyle{IEEEtrans}

\appendix

%(completely positive trace preserving map)

\subsection{ Proof of Lemma \ref{lemma-thermal}.}

To compress the $n$-fold thermal state $(\rho^{{\rm thm}}_{\beta})^{\otimes n}$ into classical memory, we can first do measurements on the state, and then encode the outcome using the smallest  possible classical memory. The state $\rho^{{\rm thm}}_{\beta}$ is fully described by the thermal parameter $\beta$, or equivalently, the average photon number $\beta/(1-\beta)$. An  estimate of $\beta/(1-\beta)$ can be obtained by measuring the photon number of the $n$ modes independently and by computing the sum.  The sum  can  be encoded  into a classical memory. To this purpose,  we divide the set of nonnegative integers into intervals, and only store the index of the interval the sum lies in.

For any $\delta>0$, we define a series of $t+1=\lfloor n^{1/2+\delta}\rfloor$ intervals $\set L_0, \dots \set L_t$ as
\begin{align*}
\set{L}_0&=\{ 0\}\\
\set{L}_i&=\{ (i-1) \, \lfloor n^{(1-\delta)/2}\rfloor+1 ,\dots, i \,  \lfloor n^{(1-\delta)/2}\rfloor   \}\quad 0<i<t\\
\set{L}_t&=\{  (t-1) \, \lfloor n^{(1-\delta)/2}\rfloor+1 ,\dots \}.
\end{align*}
%We also denote by $$\set{Ini}=\{0,\lfloor \sqrt{n}/\log n \rfloor, \lfloor 2\sqrt{n}/\log n\rfloor, \dots\}$$ the collection of all initials of these intervals.
For any non-negative integer $m$, we denote by $i(m)$ the index of the interval containing $m$, i.e. $m\in\set{L}_{i(m)}$.

To design the compression protocol, we notice that the $n$-fold thermal state can be written in the form
\begin{align}
(\rho^{{\rm thm}}_{\beta})^{\otimes n}=(1-\beta)^n\sum_{\vec{m}}\beta^{|\vec{m}|}|\vec{m}\>\<\vec{m}|,
\end{align}
where $|\vec{m}\>=|m_1\>\otimes \cdots\otimes|m_n\>$ is the photon number basis of $n$ modes and $|\vec{m}|:=m_1+\cdots+m_n$.
The compression protocol runs as follows:
\begin{itemize}
\item \emph{Encoder.} First perform projective measurement in the photon number basis of $n$ modes, which yields an $n$-dimensional vector $\vec{m}$.
%For convenience, we define a map $f$ which takes any $\vec{m}$ to the initial of the interval containing $|\vec{m}|$, formally defined as
%\begin{align*}
%f: \vec{m}\to m_{\rm ini}\in\set{Ini}\quad{\rm s.t.}\quad m_{\rm ini}\in\set{L}_{i(|\vec{m}|)}.
%\end{align*}
Then compute $i(|\vec{m}|)$ and encode it into a classical memory.
The encoding channel can be represented as
\begin{align*}
\map{E}^{{\rm thm}}_{n,\delta}(\rho)&:=
\sum_{\vec{m}}\<\vec{m}|\rho|\vec{m}\>\ |i(|\vec{m}|)\>\<i(|\vec{m}|)|.
\end{align*}
\item \emph{Decoder.} Read the integer $i$ from the memory. If $i\ge t$ prepare a fixed state $|\vec{t}\>\<\vec{t}|$ (defined below); if $i<t$ perform random sampling in the interval $\set{L}_i$. For each outcome $\hat{m}$ of the sampling, prepare the $n$-mode state
$${n+\hat{m}-1\choose \hat{m}}^{-1}\sum_{\vec{m}:|\vec{m}|=\hat{m}}|\vec{m}\>\<\vec{m}|.$$
 Then the decoding channel can be represented as
\begin{align*}
&\map{D}^{{\rm thm}}_{n,\delta}\left(|i\>\<i|\right):=\left\{\begin{matrix}\sum_{\vec{m}:|\vec{m}|\in \set L_i}
\frac{  |\vec{m}\>\<\vec{m}|}{|\set{L}_i |{n+|\vec{m}|-1\choose |\vec{m}|}}&\quad i<t\\
|\vec{t}\>\<\vec{t}| & \quad   i= t\end{matrix}\right.\\
&\vec{t}=((t-1) \, \lfloor n^{(1-\delta)/2}\rfloor+1,\dots,(t-1) \, \lfloor n^{(1-\delta)/2}\rfloor+1).
\end{align*}
\end{itemize}
It is straightforward from definition that the protocol uses $\log (t+1)=(1/2+\delta)\log n+o(1)$ classical bits. What remains is to bound the error of the protocol. First, we notice that the recovered state is
\begin{align}
&\map{D}^{{\rm thm}}_{n,\delta}\circ\map{E}^{{\rm thm}}_{n,\delta}\left[(\rho^{{\rm thm}}_{\beta})^{\otimes n}\right]=\sum_{\vec{m}}\left(\rho_{\beta,n}\right)_{\vec{m}}|\vec{m}\>\<\vec{m}|\\
&\left(\rho_{\beta,n}\right)_{\vec{m}}=\left\{\begin{matrix}(1-\beta)^n\beta^{|\vec{m}|}\sum_{m\in\set{L}_{i(|\vec{m}|)}}\frac{\beta^{m-|\vec{m}|}{n+m-1\choose m}}{|\set{L}_{i(|\vec{m}|)}|{n+|\vec{m}|-1\choose |\vec{m}|}}& \quad i(|\vec{m}|)<t\\
\\
\sum_{\vec{m}':i(|\vec{m}'|)=t}(1-\beta)^n\beta^{|\vec{m}'|}&\quad \vec{m}=\vec{t}\\
\\
0&\quad \text{else}.
\end{matrix}\right.
\end{align}

We choose $\set{S}$ as the minimal set satisfying $i)$ $\set{S}$ is a union of several intervals chosen from the set $\{\set{L}_i\}$ and ii)
\begin{align}
\set{S}\supset \left[\frac{\beta n}{1-\beta}-n^{(1-\delta)/2},\frac{\beta n}{1-\beta}+n^{(1-\delta)/2}\right].
\end{align}
Apparently, $\set{S}\subset[0,n^{1+\delta/2}]$ for large enough $n$.
Then the error can be bounded as
\begin{align*}
\epsilon_{\rm thm}&=\frac12\left\|\map{D}^{{\rm thm}}_{n,\delta}\circ\map{E}^{{\rm thm}}_{n,\delta}\left[(\rho^{{\rm thm}}_{\beta})^{\otimes n}\right]-(\rho^{{\rm thm}}_{\beta})^{\otimes n}\right\|_1\\
&=\frac12\sum_{\vec{m}}\left|(1-\beta)^n\beta^{|\vec{m}|}-\left(\rho_{\beta,n}\right)_{\vec{m}}\right|\\
&\le \sum_{\vec{m}:|\vec{m}|\not\in\set{S}}(1-\beta)^n\beta^{|\vec{m}|}+\frac12\left[\sum_{\vec{m}:|\vec{m}|\in\set{S}}(1-\beta)^n\beta^{|\vec{m}|}\right]  \\   
\nonumber &   \qquad \qquad  \qquad\times \max_{\tiny\begin{matrix}m'\in\set{L}_{i}\\\set{L}_i\cap\set{S}\not=\emptyset\end{matrix}}\left|\sum_{m\in\set{L}_i}\frac{\beta^{m}{n+m-1\choose m}}{|\set{L}_i|\beta^{m'}{n+m'-1\choose m'}}-1\right|\\
\nonumber & \qquad \qquad \qquad +\sum_{\vec{m}':i(|\vec{m}'|)=t}(1-\beta)^n\beta^{|\vec{m}'|}\\
&\le \sum_{\vec{m}:|\vec{m}|\not\in\set{S}}(1-\beta)^n\beta^{|\vec{m}|}+\frac12\left[\sum_{\vec{m}:|\vec{m}|\in\set{S}}(1-\beta)^n\beta^{|\vec{m}|}\right]  \\
\nonumber & \qquad \qquad \qquad \times \max_{\tiny\begin{matrix}m,m'\in\set{L}_{i}\\\set{L}_i\cap\set{S}\not=\emptyset\end{matrix}}\left|\frac{\beta^{m}{n+m-1\choose m}}{\beta^{m'}{n+m'-1\choose m'}}-1\right|  \\
\nonumber &  \qquad \qquad \qquad +O\left((1-\beta)^n \beta^n\right)\\
&\le \sum_{\vec{m}:|\vec{m}|\not\in\set{S}}(1-\beta)^n\beta^{|\vec{m}|}  \\
&  \qquad+\frac12\max_{\tiny\begin{matrix}m,m'\in\set{L}_{i}\\\set{L}_i\cap\set{S}\not=\emptyset\end{matrix}}\left|\frac{\beta^{m}{n+m-1\choose m}}{\beta^{m'}{n+m'-1\choose m'}}-1\right|  \\
& \qquad +O\left((1-\beta)^n \beta^n\right).
\end{align*}
On one hand, we notice that $|\vec{m}|$ is the sum of $n$ i.i.d. random variables with geometric distribution $\{(1-\beta)\beta^i\}_{i=0}^{\infty}$ and thus, by Central Limit Theorem, the first error term scales as
\begin{align*}
\sum_{\vec{m}:|\vec{m}|\not\in\set{S}}(1-\beta)^n\beta^{|\vec{m}|}&=O\left[ \erfc\left(\frac{n^{\delta/2}(1-\beta)}{\sqrt{2\beta}}\right)\right]=e^{-\Omega(n^{\delta})}
\end{align*}
where $\erfc(\delta):=(2/\pi)\int_\delta^{\infty}e^{-s^2}\d s$ is the complementary error function. On the other hand, in the second error term, $m$ and $m'$ are in the same order as $n$, so the second term can be bounded as
\begin{align*}
&\max_{\tiny\begin{matrix}m,m'\in\set{L}_{i}\\\set{L}_i\cap\set{S}\not=\emptyset\end{matrix}}\left|\frac{\beta^{m}{n+m-1\choose m}}{\beta^{m'}{n+m'-1\choose m'}}-1\right| \\
&\quad =\max_{\tiny\begin{matrix}m,m'\in\set{L}_{i}\\\set{L}_i\cap\set{S}\not=\emptyset\end{matrix}}\left|\beta^{m-m'}\frac{(n+m-1)\cdots(n+m')}{m\dots(m'+1)}-1\right|\\
& =\max_{\tiny\begin{matrix}m,m'\in\set{L}_{i}\\\set{L}_i\cap\set{S}\not=\emptyset\end{matrix}}\left|\left(\frac{\beta m+\beta n}{m}\right)^{m-m'}\left[1+O\left(\frac{|\set{L}_i|^2}{n}\right)\right]-1\right|\\
&=\left|\left[1+O\left(n^{-\delta}\right)\right]\left[1+O\left(n^{-\delta}\right)\right]-1\right|\\
&=O\left(n^{-\delta}\right).
\end{align*}
Therefore, we have proved Eq. (\ref{error-thermal}).

\subsection{Derivation of Eq. (\ref{amp-output})}\label{app-amp}
%For completeness, here we provide a proof of the relation
% \begin{align*}
%&\map{A}^{\gamma}\left(\rho_{\alpha,\beta}\right)=\rho_{\sqrt{\gamma}\alpha,\beta'}\qquad \beta'=\frac{\beta+\gamma-1}{\gamma}.
%\end{align*}
Eq. (\ref{amp-output}) is a standard result in quantum optics. Here we provide its derivation for the benefit of those readers who may be less familiar with this area.  

Note that the amplifier of Eq. (\ref{amplifier}) can be represented as
\begin{align}\label{amplifier-2}
\map{A}^{\gamma}(\rho)&=\Tr_B\left[S_{\cosh^{-1}(\sqrt{\gamma})}(\rho\otimes |0\>\<0|_B)S_{\cosh^{-1}(\sqrt{\gamma})}^\dag\right]
\end{align}
with $S_r:=e^{r (\hat{a}^\dag \hat{b}^\dag-\hat{a}\hat{b})}$. The unitary $S_r$ satisfies $S_r\hat{a}S_r^\dag=(\cosh r)\hat{a}-(\sinh r)\hat{b}^\dag$ (cf. Eq. (B8) of \cite{caves1982quantum}), which immediately implies the relation $S_rD_\alpha=D_{(\cosh r)\alpha}S_r$. Hence, we have
\begin{align}
\map{A}^{\gamma}\circ\map{D}_{\alpha}=\map{D}_{\sqrt{\gamma}\alpha}\circ\map{A}^{\gamma}
\end{align}
for any $\alpha\in\C$ and $\gamma\ge1$.   In particular, when  the amplifier is applied to displaced thermal states, one has the relation 
\begin{align}\label{commute-amp}
\map{A}^\gamma(\rho_{\alpha,\beta})=\map{D}_{\sqrt{\gamma}\alpha}\circ\map{A}^\gamma(\rho^{\rm thm}_\beta) \, .
\end{align}
To prove Eq. (\ref{amp-output}), it only  remains to show the identity  $\map{A}^\gamma(\rho^{\rm thm}_\beta)=\rho^{\rm thm}_{\beta'}$ with $\beta'$ as in Eq. (\ref{amp-output}).  This equality, which is  standard  in quantum optics, can be proven by observing that  every  thermal state can be generated from the vacuum through the action of the Gaussian additive noise channel $\map N_x$, defined as 
\begin{align}\label{gaussian-def}
\map{N}_x(\rho)=\int \d^2\mu\,\left(\frac{x}{\pi}\right)e^{-x|\mu|^2}\,D_\mu\,\rho D_\mu^\dag \, .
\end{align}
Specifically, it is easy to verify the relation 
\begin{align}
\rho_{\beta}^{\rm thm}= \map{N}_{\frac{1-\beta}{\beta}}(|0\>\<0|),
\end{align}
valid for every $\beta\in  (0,1]$.   Then, using Eq. (\ref{commute-amp}), we have
\begin{align}
\map{A}^\gamma(\rho^{\rm thm}_\beta)&=\map{A}^\gamma\circ\map{N}_{\frac{1-\beta}{\beta}}(|0\>\<0|)=\map{N}_{\frac{1-\beta}{\beta\gamma}}\circ\map{A}^\gamma(|0\>\<0|).
\end{align}
It is straightforward to verify that $\map{A}^\gamma(|0\>\<0|)$ is a thermal state; specifically, one has
\begin{align}
\map{N}_{\frac{1}{\gamma-1}}(|0\>\<0|) \, .
\end{align}   Hence, we have
\begin{align}
\map{A}^\gamma(\rho^{\rm thm}_\beta)&=\map{N}_{\frac{1-\beta}{\beta\gamma}}\circ\map{N}_{\frac{1}{\gamma-1}}(|0\>\<0|)\nonumber\\
&=\map{N}_{\frac{1-\beta}{\gamma-1+\beta}}(|0\>\<0|)\nonumber\\
&=\rho^{\rm thm}_{\beta'}\qquad \beta'=\frac{\beta+\gamma-1}{\gamma}\label{ABCD2},
\end{align}
the second equality following  from the composition property  of the additive Gaussian noise
\begin{align*}
\map{N}_x\circ \map{N}_y=\map{N}_{\frac{xy}{x+y}}\qquad \forall\, x, y>0.
\end{align*}
Indeed, the above equation can be derived using the definition of the additive Gaussian noise (\ref{gaussian-def}):
\begin{align*}
\map{N}_x\circ \map{N}_y(\rho)&=\int\d^2\mu\int\d^2\nu\,\left(\frac{xy}{\pi^2}\right)e^{-x|\mu|^2-y|\nu|^2}\,D_{\mu+\nu}\,\rho D_{\mu+\nu}^\dag\\
&=\int\d^2\mu\int\d^2\alpha\,\left(\frac{xy}{\pi^2}\right)e^{-x|\mu|^2}\cdot e^{-y|\alpha-\mu|^2}\,D_{\alpha}\,\rho D_{\alpha}^\dag\\
&\qquad\qquad\alpha:=\mu+\nu\\
&=\int\d^2\mu\left(\frac{x+y}{\pi}\right)e^{-(x+y)\left|\mu-\frac{y}{x+y}\alpha\right|^2}  \\
&  \qquad \times \int\d^2\alpha\,\left[\frac{xy}{(x+y)\pi}\right]e^{-\frac{xy}{x+y}|\alpha|^2}\,D_{\alpha}\,\rho D_{\alpha}^\dag\\
&=\map{N}_{\frac{xy}{x+y}}(\rho).
\end{align*}
Combining Eqs. (\ref{commute-amp}) and (\ref{ABCD2}) gives Eq. (\ref{amp-output}) as desired.

\subsection{Proof of Eq. (\ref{thermal-property})}\label{app-trivial}
The distance between 
\begin{align*}
\left\|\rho^{{\rm thm}}_{\beta'}-\rho^{{\rm thm}}_{\beta}\right\|_1&\le\frac{2|\beta'-\beta|}{(1-\beta')^2}+O(|\beta'-\beta|^2) \, .
\end{align*}
By definition
\begin{align*}
\left\|\rho^{{\rm thm}}_{\beta'}-\rho^{{\rm thm}}_{\beta}\right\|_1&=\sum_{j=0}^{\infty}|(1-\beta)\beta^j-(1-\beta')(\beta')^j|\\
&\le\sum_{j=0}^{\infty}\left[|\beta^j-(\beta')^j|+|\beta^{j+1}-(\beta')^{j+1}|\right]\\
&=2\sum_{j=0}^{\infty}|\beta^j-(\beta')^j|+O(|\beta-\beta'|^2).
\end{align*}
Extracting $|\beta-\beta'|$ from the first term on the r.h.s. of the last equality and focusing on the first order, we get
\begin{align*}
\left\|\rho^{{\rm thm}}_{\beta'}-\rho^{{\rm thm}}_{\beta}\right\|_1&\le2|\beta-\beta'|\sum_{j=0}^{\infty}(j+1)\beta^{j}+O(|\beta-\beta'|^2)\\
&=\frac{2|\beta-\beta'|^2}{(1-\beta)^2}+O(|\beta-\beta'|^2)
\end{align*}
as desired.

\subsection{ Proof of Lemma \ref{lemma-truncation}.} \label{proof-lemma-truncation}
For any input state $\rho$ we have
\begin{align*}
\epsilon(\rho)&\le\frac12\left\|P_K\rho P_K-\rho\right\|_1+\frac12[1-\Tr(\rho P_K)]\\
&\le \frac12\left[2\sqrt{[1-\Tr(\rho P_K)]}+1-\Tr(\rho P_K)\right]\\
&\le \frac32\sqrt{1-\Tr(\rho P_K)}.
\end{align*}
The second inequality came from the gentle measurement lemma \cite{gentle,wilde}. Substituting $\rho_{\alpha,\beta}=D_\alpha\rho^{{\rm thm}}_{\beta} D^\dag_\alpha$ into the above bound, we express the truncation error for $\rho_{\alpha,\beta}$ as
\begin{align}
\epsilon(\rho_{\alpha,\beta})&\le \frac32\sqrt{1-\Tr\left[D_{\alpha}\rho^{{\rm thm}}_{\beta} D^\dag_{\alpha}P_K\right]}\label{eN-aaa}.
\end{align}
We now bound the trace part in the right hand side of the last inequality as
\begin{align}
\nonumber1-\Tr\left[D_{\alpha}\rho^{{\rm thm}}_{\beta} D^\dag_{\alpha}P_\alpha\right]=&\sum_{k\le K}\sum_{l=0}^{\infty}(1-\beta)\beta^l |\<l|D_{\alpha}|k\>|^2\\
\nonumber\le&\max_{l\le l_0}\sum_{k>K} |\<l|D_{\alpha}|k\>|^2+\sum_{l\ge l_0}^{\infty}(1-\beta)\beta^l \\
\le&\max_{l\le l_0}\sum_{k>K} |\<l|D_{\alpha}|k\>|^2+\beta^{l_0}\label{error-mixed-coherent-1}
\end{align}
Here we set $l_0=K^{x/8}$. Notice that $|\<l|D_{\alpha}|k\>|^2$ is the photon number distribution of a displaced number state \cite{displaced-number}, which can be bounded as
\begin{align*}
|\<l|D_{\alpha}|k\>|^2=&\frac{e^{-|\alpha|^2}(|\alpha|^2)^{k+l}}{k!l!}\left|\sum_{i=0}^{\min\{k,l\}} \frac{k!l!(-|\alpha|^2)^{-i}}{i!(k-i)!(l-i)!}\right|^2\\
\le&\frac{e^{-|\alpha|^2}(|\alpha|^2)^{k+l}}{k!l!}\left|\sum_{i=0}^k{k\choose i} \left(\frac{l}{|\alpha|^2}\right)^i\right|^2\\
=&\frac{e^{-|\alpha|^2}|\alpha|^{2l}}{k!l!}\left(\frac{l+|\alpha|^2}{|\alpha|}\right)^{2k}.
\end{align*}
Then we can bound the first term in (\ref{error-mixed-coherent-1}) as
\begin{align*}
\max_{l\le l_0}\sum_{k>K} |\<l|D_{\alpha}|k\>|^2\le&\max_{l\le l_0}\frac{|\alpha|^{2l}}{l!}\sum_{k>K}\frac{e^{-|\alpha|^2}}{k!}\left(\frac{l+|\alpha|^2}{|\alpha|}\right)^{2k}\\
=&\max_{l\le l_0}\frac{e^{2l+\frac{l^2}{|\alpha|^2}}|\alpha|^{2l}}{l!}\sum_{k>K}\mathbf{Pois}_{\lambda_\alpha}(k),
\end{align*}
where $\mathbf{Pois}_{\lambda_\alpha}(k)$ is the Poisson distribution with mean $\lambda_\alpha=(l+|\alpha|^2)^2/|\alpha|^2$. Notice that $[\lambda_\alpha-K^{1/2+x/2},\lambda_\alpha+K^{1/2+x/2}] \subseteq [0,K]$ and thus we have
\begin{align*}
\max_{l\le l_0}\sum_{k>K} |\<l|D_{\alpha}|k\>|^2\le&\max_{l\le l_0}\frac{|\alpha|^{2l}e^{2l+\frac{l^2}{|\alpha|^2}}}{l!}\sum_{|k-\lambda_\alpha|>K^{1/2+x/2}}\mathbf{Pois}_{\lambda_\alpha}(k)\\
=&\max_{l\le |\alpha|^{x/4}}\frac{|\alpha|^{2l}e^{2l+\frac{l^2}{|\alpha|^2}}}{l!}e^{-\Omega(K^{x/2})}\\
=&e^{-\Omega(K^{x/4})}.
\end{align*}
having used the tail bound for Poisson distribution and $l_0=K^{x/8}$.
Substituting the above bound into Eq. (\ref{error-mixed-coherent-1}), we get
\begin{align*}
\sqrt{1-\Tr\left[D_{\alpha}\rho^{{\rm thm}}_{\beta} D^\dag_{\alpha}P_\alpha\right]} \le& \beta^{\Omega(K^{x/8})}+e^{-\Omega(K^{x/4})}.
\end{align*}
Finally, substituting the above inequality into Eq. (\ref{eN-aaa}), we can bound the error as
\begin{align*}
\epsilon(\rho_{\alpha,\beta})=\beta^{\Omega(K^{x/8})}+e^{-\Omega(K^{x/4})}.
\end{align*}

\end{document}